\documentclass[11pt]{article}

\usepackage[margin = 1in]{geometry}

\usepackage{enumerate}
\usepackage{filecontents}
\usepackage{amsmath,amssymb,amsthm,hyperref,color}
\usepackage{filecontents}
\hypersetup{colorlinks=true,citecolor=blue, linkcolor=blue, urlcolor=blue}

\def\E{\mathbf{E}}
\def\e{\mathrm{e}}

\def\pib{\boldsymbol{\pi}}

\def\Tb{\mathbb{T}}
\def\Tc{\mathcal{T}}
\def\tb{\mathbf{t}}
\def\xb{\mathbf{x}}
\def\yb{\mathbf{y}}
\def\boxcore{\mathsf{box}\mbox{-}\mathsf{core}}

\def\norm#1{\left\|#1\right\|}
\def\dist#1#2{{\mathrm{dist}(#1,#2)}}

\newtheorem{theorem}{Theorem}
\newtheorem{lemma}[theorem]{Lemma}

\newtheorem{corollary}[theorem]{Corollary}

\newtheorem{definition}[theorem]{Definition}

\newcommand{\poly}{\mathrm{poly}}
\newcommand{\eps}{\epsilon}

\newcommand{\pb}{\mathbf{p}}

\newcommand{\Tmix}{T_{\mathrm{mix}}}

\newcommand{\Gc}{\mathcal{G}}

\def\diam{\mathrm{diam}}
\newcommand{\integers}{\mathbb{Z}}

\newcommand{\partialv}{\partial_{\mathtt{v}}}

\newtheorem*{lemarbexp}{Lemma~\ref{lem:arbexp}}
\newtheorem*{lemnumcon}{Lemma~\ref{lem:numcon}}
\newtheorem*{lemconrltwo}{Lemma~\ref{lem:conrltwo}}

\newtheorem*{lemarbexpreqtwo}{Lemma~\ref{lem:arbexpreqtwo}}
\newtheorem*{lemavgdegtwo}{Lemma~\ref{lem:avgdegtwo}}
\newtheorem*{lemconreqtwo}{Lemma~\ref{lem:conreqtwo}}
\newtheorem*{lembigexpansiontwo}{Lemma~\ref{lem:bigexpansiontwo}}

\begin{document}

\title{Random Walks on Small World Networks}

\author{Martin E. Dyer\thanks{University of Leeds, UK. Email: m.e.dyer@leeds.ac.uk.  Supported by EPSRC research grant EP/M004953/1. }  
\and Andreas Galanis\thanks{
University of Oxford, UK.  The research leading to these results has received funding from the European Research Council under 
  the European Union's Seventh Framework Programme (FP7/2007-2013) ERC grant agreement no.\ 334828. The paper 
  reflects only the authors' views and not the views of the ERC or the European Commission. 
  The European Union is not liable for any use that may be made of the information contained therein.} 
  \and Leslie Ann Goldberg$^\dag$
\and Mark Jerrum\thanks{Queen Mary, University of London, UK. Supported by EPSRC grant EP/N004221/1 and EPSRC grant EP/I011935/1.} \and Eric Vigoda\thanks{Georgia Institute of Technology, USA.  Email: ericvigoda@gmail.com. Research supported in part by NSF grants CCF-1617306 and CCF-1563838.
}
}

\date{February 26, 2020}

\maketitle

\begin{abstract}
We study the mixing time of random walks on small-world networks modelled as follows: starting with the 2-dimensional periodic grid, each pair of vertices $\{u,v\}$ with distance $d>1$ is added as a ``long-range" edge with probability proportional to $d^{-r}$, where $r\geq 0$ is a parameter of the model. Kleinberg studied a close variant of this network model and proved  that the (decentralised) routing time is $O((\log{n})^2)$ when $r=2$
and $n^{\Omega(1)}$ when $r\neq 2$.  Here, we prove that the random walk also undergoes a phase transition at $r=2$, but in this case the phase transition is  of a different form. We establish that the mixing time is $\Theta(\log n)$ for $r<2$, $O((\log n)^4)$ for $r=2$ and $n^{\Omega(1)}$ for $r>2$. 
\end{abstract}

\thispagestyle{empty}

\newpage

\setcounter{page}{1}

\section{Introduction}

A small-world network is a graph with small average degree and short diameter.
Such networks were originally designed to model the social phenomenon known
as six degrees of separation, which was popularized by Milgram~\cite{Milgram1,Milgram2} in
his experiment of routing letters.

Kleinberg~\cite{Kleinberg} introduced an intriguing model of a small-world network.
Starting with the $2$-dimensional grid, each vertex $v$ adds one, possibly long-range, directed edge $(v,w)$.
The probability that the edge from~$v$ is $(v,w)$ is
 $\dist{v}{w}^{-r}/Z$ where $r\geq 0$ is a
parameter of the model, $\dist{v}{w}$ is the grid distance between~$v$ and~$w$ 
and $Z = \sum_{y\in V} \dist{v}{y}^{-r}$ is the appropriate normalising factor.  

Kleinberg proved that the routing problem on this network has
an  interesting phase transition  at $r=2$: when $r=2$ 
there is a decentralised routing algorithm with expected delivery
time $O((\log n)^2)$, whereas for all $r\neq 2$ 
every decentralised routing algorithm has expected delivery time $\Omega(n^c)$ where
$c$ is a constant that grows with $|r-2|$.

Our goal in this paper is to analyse the behaviour of the random walk
on such a small-world network with respect to the parameter $r$.
Routing algorithms imply bounds on the 
diameter of the underlying graph but do not yield significant insight into further
macroscopic properties.  In contrast, 
the mixing properties of the random walk yields insights into the macroscopic
connectivity of the underlying graph, such as the conductance~\cite{JS,LP},
and have connections with gossip-style protocols~\cite{Gossip1,Gossip2}.

We focus our attention on a natural variant of Kleinberg's model.  Our model
is similar to the model of Abraham et al.~\cite{ACKS} who study the critical case $r=2$,
and is closely related to long-range percolation in the case $r>2$ as described below. 
 Instead of adding \emph{exactly} one \emph{directed} long-range edge from each vertex, in our model edges are \emph{undirected} and added independently, so that the \emph{expected} number of long-range edges incident to each vertex is one. 
Also, we consider the underlying graph to be a torus (periodic grid) so
that the underlying graph is vertex transitive and the normalising factor $Z$ is identical for all vertices. The fact that edges are undirected is technically convenient, as is the fact that  the presence/absence of an edge is mutually independent for each pair of vertices.
Studying the random walk on the directed version of the model is a challenging open
question (see Section~\ref{sec:discussion}) since
the directed model is non-reversible, and hence 
even understanding the stationary distribution is difficult.

The details of our model are as follows.
Following Kleinberg, we have a parameter $r\geq 0$. 
The model $\Gc_{n,r}$ is obtained by taking the 
$2$-dimensional\footnote{For simplicity we focus on the 2-dimensional case
where the underlying graph is a   (periodic) 2-dimensional   grid.
Kleinberg generalised his result to the case where the underlying grid is $d$-dimensional and
our results can be generalised similarly --  see Section~\ref{sec:discussion}.}
 torus (periodic grid graph) with side length $(2n+1)$ centered at the origin.
Independently, 
for every pair of vertices $v,w$ which are not connected by a torus edge, we add the (undirected) edge $\{v,w\}$
with probability $\dist{v}{w}^{-r}/Z$, where the normalising 
factor~$Z$ is given by $Z = \sum_{y\in V; y\neq v} \dist{v}{y}^{-r}$
and the distance $\dist{v}{w}$ is the graph distance between $v$ and $w$ in the
original torus.  
The normalising factor~$Z$ is important --- this is what ensures that the expected number of long-range edges
adjacent to a vertex~$v$ is~$\Theta(1)$, as in Kleinberg's model.
We refer to  the random graph $\Gc_{n,r}$ as the ``small-world network''.

When $r>2$ our model is closely related to long-range percolation (LRP).
In LRP there are two parameters $r,\beta>0$.
Starting with the infinite $d$-dimensional lattice~$\integers^d$, 
for every pair of vertices $v,w$,
the (undirected) edge $\{v,w\}$ is added with probability $1-\exp(-\beta \,\dist{v}{w}^{-r})$
which is asymptotically $\beta \,\dist{v}{w}^{-r}$.
The main difference between LRP and our small-world model is the absence of the normalising
factor~$Z$ in the edge-addition probability. The absence of this normalising factor
  manifests itself in various ways --- perhaps most strikingly 
  by varying
 the average number of long-range neighbours of a vertex (as $r$ changes).
 Unlike our small-world model where the average degree is always $\Theta(1)$,
 in LRP it is proportional to~$Z$, and hence (as we will see in Lemma~\ref{lem:factor}
 for the asymptotics of $Z$)
the average degree in LRP is 
      $n^{\Omega(1)}$ for $r<d$, $\Theta(\log n)$ for $r=d$ and $O(1)$ for $r>d$.
      Thus, in 2-dimensions for $r>2$, our model is quite similar to LRP, but
      the models are quite different for $r\leq 2$.

To further understand the properties of small-world networks and how they
vary with~$r$, 
we study the mixing time of the lazy random walk on~$\Gc_{n,r}$. 
Let $X_t$ denote the vertex 
that the walk visits  at time $t$. 
With probability $1/2$, we set $X_{t+1}=X_t$; otherwise,   $X_{t+1}$ is
chosen  to be a neighbour of $X_t$, selected uniformly at random from the set
of all neighbours.  The lazy random walk is an 
ergodic Markov chain; we let $\pi$ denote its unique stationary distribution.
The mixing time~$\Tmix$ is the minimum number of steps $T$, from the worst initial state $X_0$, to guarantee
that the distribution of $X_T$ is within total variation distance $\leq 1/4$ of the stationary
distribution $\pi$, see Section~\ref{sec:lazyrandomwalk} for more thorough definitions.
We prove that the mixing time of the lazy random walk on the
small-world network undergoes a phase transition at $r=2$. 

\begin{theorem}\label{thm:main}
Let $r\geq 0$ and $n$ be a positive integer. With probability $1-O(1/n)$ over the choice of
the small-world network~$\Gc_{n,r}$, the lazy random walk 
on this network satisfies:
\begin{equation*}
\Tmix = \begin{cases}
\Theta(\log{n}) & \mbox{ if } r<  2 \\
O((\log{n})^4) & \mbox{ if } r=  2 \\
n^{\Omega(1)} & \mbox{ if } r>2. \\
\end{cases}
\end{equation*}
\end{theorem}
 
Note that the phase transition in Theorem~\ref{thm:main} is different from the phase transition obtained by Kleinberg for the greedy routing algorithm. When $r<2$, the greedy routing algorithm is slow whereas the lazy random walk is as fast as possible, with mixing time $O(\log n)$.

Our main technical challenge is getting sharp upper bounds on the mixing time for the
case $r<2$.  We believe that $\Theta((\log{n})^2)$ for $r=2$ is the correct answer and
hence obtaining an upper bound that is asymptotically smaller for $r<2$ is 
especially interesting.  A challenging open problem is to show a separation
for the cases $r<2$ vs. $r=2$, however we do not have a good approach for
proving a non-trivial lower bound, see
Section~\ref{sec:discussion} for more discussion.

\subsection{Related works}\label{sec:related}

The mixing time of random walks in sparse random graphs is a well-studied subject.  Our proof of Theorem~\ref{thm:main} adapts methods that have been developed for analysing random walks in 
the Erd\"{o}s-R\'{e}nyi random graph $G(n,p)$.  
Fountoulakis and Reed \cite{FRFR} proved that, for any $d>1$,
the mixing time of the lazy random walk on the largest component of  $G(n,d/n)$ 
is $\Theta((\log n)^2)$. 
See \cite{BKW} for an alternative proof of this result,
 and \cite{DLP,NP} for mixing-time analysis 
 when the parameter~$p$ is chosen so that $G(n,p)$ is in,  or close to, the critical window for the emergence of the giant component.
The proof technique of Fountoulakis and Reed,
based on analysing  the conductance of connected sets,   
is the basis of 
the approach that we use to obtain the upper bounds in Theorem~\ref{thm:main} (see Section~\ref{sec:reerg}). 
 
The random graph $G(n,p)$ is constructed by starting with
the most straightforward underlying graph --- the empty graph with $n$ vertices --- and 
turning each non-edge into an edge independently with probability~$p$.
Of course, it is also possible to start with a different
underlying graph.
A well-known example of a
random graph model of this type is the small-world model of Newman and Watts~\cite{NW}. 
In its simplest form, the
Newman-Watts random graph is the random graph obtained by starting from a cycle of length $n$ and then, independently, adding each non-edge with
probability $c/n$ for some constant $c>0$. 
Addario-Berry and Lei \cite{AL} study the mixing time of the lazy random walk in the Newman-Watts random graph  
and prove that, for any constant $c>0$, the mixing time is $O((\log{n})^2)$, matching
the lower bound of $\Omega((\log{n})^2)$ of Durrett \cite{Durrett}. 
The proof technique of Addario-Berry and Lei  is based on 
bounding the number of connected sets, and we 
adapt the bounding technique of \cite{AL}  to our setting in 
Section~\ref{sec:connectedsets}.  
Similar random graph models with other underlying graphs are considered in \cite{flaxman,flaxmanfrieze,KRS}. 
For example, Krivelevich, Reichman and Samotij~\cite{KRS} 
study the mixing time of the lazy random walk on a random graph
which is formed by adding edges independently to an arbitrary connected graph --- they show upper bounds of $O((\log n)^2)$.
 
The models discussed in the previous paragraph have the property that each non-underlying edge
is added with the same probability, $p$.
There are also random graph models in which the probability that a pair $\{u,v\}$ of
vertices is added as an edge depends on  
the underlying graph distance $\dist{u}{v}$.
Our small-world model 
is of this type and so
is   Kleinberg's small-world model \cite{Kleinberg, kleinberg2006}, which  we have already discussed.  
Nguyen and Martel \cite{NMdiamb, NMdiam} study random graphs in Kleinberg's model 
and show that the diameter undergoes the following phase transitions in $d$-dimensions: $\Theta(\log n)$ for $r\leq d$, $(\log n)^{\Theta(1)}$ for $d<r<2d$ and $n^{\Omega(1)}$ for $r>2d$. 

As mentioned earlier, long-range percolation (LRP) is another random graph model in which an
underlying graph is augmented by adding each new edge $\{u,v\}$ with probability 
 approximately   $\beta \,\dist{u}{v}^{-r}$. 
There is an an extensive body of work on the long-range percolation model --- see \cite{ben2008,CS1} for the mixing time, \cite{Aizenman,Schul83,newman1986, aizenman1986rf} for percolation, and \cite{BB,benkest04, CGS,NBerger,biskup2004, biskup2011,dingsly2013, bisjin2017} for diameter. 
For the 1-dimensional model, Benjamini, Berger, and Yadin \cite{ben2008} 
showed in LRP that the relaxation time of the simple random walk on the infinite component
undergoes a phase transition at $r=2$: 
it is $n^{r-1}(\log n)^{\Theta(1)}$ for $1<r<2$ and $\Omega(n^2)$ for $r>2$.
More generally, for $d\geq 1$ dimensions, Crawford and Sly~\cite{CS1} 
proved that the relaxation 
time for $d<r<\min\{d+2,2d\}$ is $n^{r-d}(\log n)^{O(1)}$.  For the 2-dimensional case,
straightforward lower bounds match these
results when $2<r<3$ (up to logarithmic factors), see Theorem~\ref{thm:main2} in Section~\ref{sec:lowout} for a more detailed discussion.

Finally, 
our small-world network   has been studied by \cite{JKRS}  in the special case $r=1$.
However, the results  of \cite{JKRS} are about bootstrap percolation, rather than
being about the mixing time of the lazy random walk, which is our concern here.

\subsection{Outline}\label{sec:outline}

To obtain Theorem~\ref{thm:main}, one of our main tasks is to bound the edge-expansion ratio $|\partial S|/\min\{|S|,|V\backslash S|\}$ for all subsets $\emptyset \subset S \subset V$, where $\partial S$ denotes the set of edges with exactly one endpoint in $S$.\footnote{Technically, we need to consider the edge expansion of subsets $S$ whose complement $V\backslash S$ has non-trivial density, i.e.,   $|V\backslash S|\geq \delta |V|$ where $\delta\in(0,1)$ is a small constant bounded away from zero ($\delta=1/100$ is sufficient).} In our setting where every vertex has bounded-degree in expectation, lower bounds on the edge-expansion ratio can be used to lower bound the so-called conductance of the random walk (see Section~\ref{sec:reerg}) and therefore obtain upper bounds on the mixing time (with some extra work);  similarly,  upper bounds  on the edge-expansion ratio can be used to obtain lower bounds on the mixing time.

The key idea is that for $r<2$ the long-range edges mostly connect far away vertices in the torus (i.e., vertices at distance $\Omega(n)$) and the edge-expansion of a set $S$ in the small-world graph is likely to be bounded by a constant. More precisely, there exists a constant $c>0$ such that for all sets with $S\subseteq V$ (and $|S|\leq 99|V|/100$) it holds that 
\begin{equation}\label{eq:4g4g45g}
\mbox{$\Pr_{\Gc_{n,r}}$}(|\partial S|\geq c |S|)\geq 1-\exp(-c|S|).
\end{equation}
A somewhat similar inequality holds for $r=2$, but now sets $S$ with $\Omega(n^2)$ vertices have worse edge-expansion ratio of order $1/\log n$, see Lemma~\ref{lem:arbexpreqtwo} for the precise statement. In contrast, for $r>2$, the long-range edges are more likely to connect vertices which are at constant distance in the torus and the expansion properties of the small-world graph become qualitatively similar to the torus. More precisely, one can find large sets $S$ with $\Omega(n^2)$ vertices whose expansion  is at most $n^{-\Omega(1)}$, see Lemma~\ref{lem:torpidconductance}. These sets with bad edge-expansion ratio give the lower bound on the mixing time in Theorem~\ref{thm:main} for $r>2$, see also Theorem~\ref{thm:main2} in Section~\ref{sec:lowout} for details (there, we also show $\Omega(\log n)$ lower bounds on the mixing time for $r\leq 2$ using a simple diameter lower bound).

Obtaining upper bounds on the mixing time for $r< 2$ on the basis of \eqref{eq:4g4g45g} is the key obstacle to obtaining Theorem~\ref{thm:main} (and similarly for $r=2$). Namely, to obtain lower bounds on the edge-expansion ratio, we would ideally want to combine the probabilistic estimates given in \eqref{eq:4g4g45g} for all sets $S\subseteq V$. The trouble is that we cannot easily combine the probabilistic estimates given in \eqref{eq:4g4g45g}  for all sets $S\subseteq V$ since the straightforward union bound fails miserably. For example, for $k=n^{2-\Omega(1)}$, there are roughly $\binom{n^2}{k}=\e^{\Omega(k\log n)}$ sets $S$ with $|S|=k$, but the event $|\partial S|\geq c |S|$ for a set $S$ with $|S|=k$ fails to hold with probability as large as $\e^{-\Omega(k)}$. To overcome the failure of the union bound, we need to reduce the number of sets $S$ under consideration.

The first idea to perform a refined union bound is to use a theorem by Fountoulakis and Reed~\cite{FR}, see Theorem~\ref{thm:FR}, which allows us to consider only sets $S$ which are connected in $G$ (i.e., the induced subgraph on $S$ is connected). The idea is that the graph $G$ has bounded average degree and therefore the number of connected sets with $|S|=k$ containing a specific vertex should be roughly $(c')^{k}$ for some constant $c'>1$ (which is a significant improvement over the roughly $\binom{n^2}{k}$ possible sets  with $|S|=k$). Indeed, we show that this is the case by adapting techniques of Addario-Berry and Lei \cite{AL} to our setting (see Lemma~\ref{lem:numcon} in Section~\ref{sec:connectedsets}). Unfortunately, the constant $c$ in \eqref{eq:4g4g45g} turns out to be roughly equal to 1, the expected number of long-range edges incident to a vertex, while the best bound we can hope to get on the constant $c'$ turns out to be roughly 20 (four times the average degree). So, we need to reduce the number of sets $S$ further.

The second idea is that we can reduce the sets $S$ under consideration by utilising the edge-expansion of the
underlying graph and, in particular, the torus (similar type of arguments have been used in \cite{AL,KRS}). The rough intuition is that the sets in the torus that have low edge-expansion are ``box-like'', i.e., unions of boxes where a box refers to a square subgraph of the torus (see Section~\ref{sec:part}). In contrast, sets in the torus that are more spread out (for example, unions of paths on alternate layers) have constant edge expansion. By considering a partition of the torus into boxes of side length roughly equal to $\ell$ (for some large enough constant $\ell$), we can reduce the number of box-like sets we need to consider. In fact,  it turns out that it is enough to consider box-like sets $S$ which are connected in $G$, whose number we can control in a manner analogous to the one we discussed for general connected sets $S$. Eventually, we are able to control the interplay of the constant $c$ in \eqref{eq:4g4g45g} with the constant controlling the logarithm of the number of connected box-like sets $S$ by adjusting the length $\ell$ of the boxes in the torus (we show that taking $\ell$ to be a sufficiently large constant suffices). 

We should also mention that we cannot use \eqref{eq:4g4g45g} to account for the expansion of sets $S$ with small cardinality ($O(\log n)$ vertices); such small sets may  have no long-range edges incident to them and their edge expansion can be as low as $1/\sqrt{\log n}$ (see Lemma~\ref{lem:conrltwo}); using the standard conductance techniques this would lead to a mixing time bound of $O((\log n)^2)$. Nevertheless, we can get the $O(\log n)$ bound in Theorem~\ref{thm:main} for $r<2$ by employing the ``average conductance" technique of Lov\'{a}sz and Kannan \cite{LKfaster}, as refined by Fountoulakis and Reed \cite{FR} (see Theorem~\ref{thm:FR} in this paper). The full argument for $r<2$ can be found in Section~\ref{sec:rl2}.

The upper bound for $r=2$ builds on similar arguments, though some modifications are needed to deal with the worse probabilistic estimates for the edge expansion of large sets $S$ (see Lemma~\ref{lem:arbexpreqtwo}). To account for these, we rely on the expansion properties of the torus more significantly by taking the side length $\ell$ of the boxes in the partition of the torus to depend moderately on $n$; we show that taking $\ell$ to be $O((\log n)^{1/2})$ suffices. This allows us to bound the edge expansion of  box-like sets accurately, but yields rougher bounds for more spread-out sets, e.g., a large box-like set together with its long-range neighbours. Instead of bounding the edge expansion of box-like sets,  we therefore bound their vertex expansion, which has the benefit of yielding  more accurate bounds on the edge expansion of spread-out sets $S$ (this shaves off a couple of $\log n$ factors in the final mixing time upper bound). The detailed proof of  the upper bound for $r=2$ can be found in Section~\ref{sec:req2} (and the lower bound in Section~\ref{sec:lowout}).

\section{Preliminaries}\label{sec:prelim}

\subsection{Definitions}
The small-world network model $\Gc_{n,r}$ is parameterised by a positive integer $n$ and a real  $r\geq 0$. Roughly, the model $\Gc_{n,r}$ is obtained by the $(2n+1)\times (2n+1)$ two-dimensional torus (periodic grid) by adding random edges independently, where the probability of adding an edge $(u,v)$ is given by a power law with parameter $r$ in the torus distance between $u$ and $v$.  

More formally, we will denote the torus by $T=(B_{n},E_n)$ where $B_n:=\{-n,\hdots,n\}^2$ and, for two vertices $u=(x_1,y_1),v=(x_2,y_2)\in B_n$, the edge $(u,v)$ belongs to $E_n$ if, either $x_1=x_2$ and $|y_1-y_2|=1$ or $2n$, or $y_1=y_2$ and  $|x_1-x_2|=1$ or $2n$. For vertices $u,v\in B_n$, let $\dist{u}{v}$ be the length of the shortest path between $u$ and $v$ in the torus. Note that the torus is a vertex transitive graph where every vertex has degree 4.  Independently, 
for every pair of distinct vertices $u,v\in B_n$ which are not adjacent in the torus, we add the \emph{long-range} edge $\{u,v\}$
with probability $\dist{u}{v}^{-r}/Z$ where the normalising factor $Z$ is such that every vertex has in expectation one long-range edge incident to it. Note that if we denote by $\rho$ the origin of the torus, then $Z$ is given by the expression 
\begin{equation}\label{eq:defZ}
Z = \sum_{w\in V;\ \dist{\rho}{w}\geq 2} \dist{\rho}{w}^{-r}.
\end{equation}

We will use $G=(V,E)\sim \Gc_{n,r}$ to indicate that $G$ is a small-world graph with parameters $n,r$ and denote the number of vertices in $G$ by $N=(2n+1)^2$.  Using Chernoff bounds, it can be proved easily that with high probability over the choice of the graph $G$ is holds that $|E|=5N/2+o(N)$, since every vertex has average degree five. Note that exactly $2N$ of these edges come from the torus; the remaining edges are long-range (random) edges. 

For a set $S\subseteq V$, we will use $\partial S$ to  denote the subset of edges in $G$ with exactly one endpoint in $S$. We will use $\partial^* S$ to  denote the subset of edges in the torus $T$ with exactly one endpoint in $S$. Note in particular that $\partial^* S\subseteq \partial S$. 

\subsection{The lazy random walk}\label{sec:lazyrandomwalk}

Let $G=(V,E)\sim \Gc_{n,r}$. We study the Markov chain corresponding to the lazy random walk on $G$. Formally, the transition matrix $P$ of the walk is defined as  follows (where $d_v$ denotes the degree of the vertex $v$ in $G$):
\begin{equation*}  P(v,w) = \begin{cases}  \frac{1}{2} & \mbox{if } w=v\\\frac{1}{2d_v} & \mbox{if } (v,w)\in E \\
 0 & \mbox{otherwise.}
 \end{cases}
\end{equation*}
Since $G$ is connected, we have that $P$ is irreducible. Due to the self-loops, we also have that $P$ is aperiodic. It follows that the lazy random walk converges to a stationary distribution. Since $G$ is undirected, we have that $P$ is reversible with respect to the stationary distribution $\pib=(\pi_v)_{v\in V}$ given by $\pi_v=d_v/(2|E|)$.  For $S\subseteq V$, we use $\pi(S)$ to denote $\sum_{v\in S}\pi_v$. 

The mixing time is the number of steps that we need to run the chain from the worst starting state to ensure that we are within total variation distance $\leq 1/4$ from the stationary distribution. Formally, for a vertex $v\in V$, observe that the vector $P^t(v,\cdot)$ gives the distribution of the random walk starting from $v$ after $t$ steps; the total variation distance between $P^t(v,\cdot)$ and $\pib$ is given by $\norm{P^t(v,\cdot)-\pib}_{\mathrm{TV}}=\frac{1}{2}\norm{P^t(v,\cdot)-\pib}_1$. The mixing time is then defined as
\[\Tmix=\min\Big\{t\in \mathbb{Z}\, \big|\, \max_{v\in V}\norm{P^t(v,\cdot)-\pib}_{\mathrm{TV}}\leq \frac{1}{4}\Big\}.\]
It is well-known (see, e.g., \cite[Section 4.5]{LP}) that after $k\Tmix$ steps, the total variation distance between the random walk and its stationary distribution is at most $(1/2)^k$.

\subsection{Bounding the mixing time using conductance}\label{sec:reerg}
To bound the mixing time of the random walk on $G=(V,E)\sim \Gc_{n,r}$, we will bound its conductance. For the lazy random walk, for a set $\emptyset \subset S\subset V$, the normalised conductance $\Phi(S)$ is given by 
\begin{equation}\label{eq:normcond}
\Phi(S)=\frac{|\partial S|}{\frac{1}{2|E|}\big(\sum_{v\in S}d_v\big)\big(\sum_{v\notin S}d_v\big)},
\end{equation}
where recall that $\partial S$ denotes the subset of edges in $G$ with exactly one endpoint in $S$.  The conductance $\Phi$ of the chain is then given by
\[\Phi=\min_{S\neq \emptyset,V} \Phi(S).\]
\begin{theorem}[{\cite{JS}}]\label{lem:mixingtime}
There exist constants $C,C'>0$ (independent of the chain) such that the mixing time $\Tmix$ of the lazy random walk on an undirected connected graph satisfies
\[\frac{C}{\Phi}\leq \Tmix\leq \frac{C'}{\Phi^2}\log(1/\pi_{\mathrm{min}}),\]
where $\pi_{\mathrm{min}}:=\min_{v\in V}\pi_v$.
\end{theorem}

A refinement of Theorem~\ref{lem:mixingtime} was given by Fountoulakis and Reed \cite{FR} (building upon work of Lov\'{a}sz and Kannan \cite{LKfaster}), which often gives a more precise upper bound on the mixing time when the conductance of small sets is small relatively to the conductance of big sets. Following \cite{FR}, we say that a set $S\subseteq V$ is connected if the graph induced by $S$ is connected. For $0\leq p\leq 1$, let 
\[\widetilde{\Phi}(p):=\min\{\Phi(S) \mid S\mbox{ is connected},\, p/2\leq \pi(S)\leq p\},\] 
and set $\widetilde{\Phi}(p)=1$ if the minimization is over an empty set. The key feature in the definition of  $\widetilde{\Phi}(p)$ is that it only considers connected sets $S$, which vastly reduces  the sets to be considered for sparse graphs.
\begin{theorem}[{\cite{FR}}]\label{thm:FR}
There exists a constant $C>0$ (independent of the chain) such that the mixing time $\Tmix$ of the lazy random walk on an undirected connected graph satisfies
\begin{equation}\label{eq:phij12}
\Tmix\leq C \sum^{\left\lceil \log_2(\pi_{\mathrm{min}}^{-1})\right\rceil}_{j=1}\bigg(\frac{1}{\widetilde{\Phi}(2^{-j})}\bigg)^2.
\end{equation}
\end{theorem}

\subsection{Normalising factor for random edges}
The following lemma gives some basic intuition for the neighbourhood structure of $G\sim \Gc_{n,r}$ by considering the asymptotics of the normalising factor $Z$ given in \eqref{eq:defZ}. 

\begin{lemma}[\cite{Kleinberg}]\label{lem:factor}
Let $r\geq 0$. Then, for the torus $T=(B_n,E_n)$, it holds that
\[Z=\begin{cases}
\Theta(n^{2-r})  & \mbox{if } r<2,\\
\Theta(\log{n}) & \mbox{if } r=2,\\
  \Theta(1) & \mbox{if } r>2.
\end{cases}\]
\end{lemma}
\begin{proof}
Observe that all the vertices of the torus are within distance $2n$ from the origin $\rho$. For $\ell=1,\hdots,2n$, let $S_\ell$ denote the vertices at distance $\ell$ from the origin. Then, 
\[|S_\ell|=4\min\{\ell,2n+1-\ell\}.\]
It follows that 
\begin{equation}
\label{eq:Z-bound}
 Z  = \sum_{\ell=2}^{n} \frac{4\ell}{\ell^r}+\sum_{\ell=n+1}^{2n} \frac{4(2n+1-\ell)}{\ell^r}.
\end{equation}
For a positive integer $m$, let $H_m:=4\sum_{\ell=1}^{m}1/\ell^{r-1}$. We then have the bounds  
\[H_n-4\leq Z\leq H_{2n},\]
where the ``$-4$'' in the first inequality is to account for the term corresponding to  $\ell=1$ and the latter inequality follows by using the inequality $2n+1-\ell\leq \ell$ to bound the terms in the second sum in \eqref{eq:Z-bound}. 

Note that for any $x>0$, we have that $\sum^n_{i=1}i^x=\Theta(n^{x+1})$, so, for $0\leq r<1$, it holds that $H_n=\Omega(n^{2-r})$ and $H_{2n}=O(n^{2-r})$, so that $Z=\Theta(n^{2-r})$, as wanted. For $r\geq 1$ and $\ell=1,\hdots,2n$, we have that  $1/(\ell+1)^{r-1}\leq \int^{\ell+1}_{\ell}\frac{1}{x^{r-1}}\,dx\leq 1/\ell^{r-1}$, and hence we have  the bounds
\begin{equation}\label{eq:Hnboundrew}
4\int^{n+1}_{1}\frac{1}{x^{r-1}}\,dx\leq H_n, \quad H_{2n}\leq 4\Big(1+\int^{2n+1}_{1}\frac{1}{x^{r-1}}\,dx\Big),
\end{equation} 
from where it follows that $Z=\Theta(n^{2-r})$ if $r<2$, $Z=\Theta(\log n)$ if $r=2$, and $Z=\Theta(1)$ if $r>2$.
\end{proof}

\subsection{Concentration bounds}

We will use the following version of the well-known Chernoff/Hoeffding inequality.
\begin{lemma}[see, e.g., {\cite[Theorem 21.6 \& Corollary 21.9]{FK}}]\label{lem:chernoff}
Suppose that $S_n=X_1+\cdots+X_n$, where $\{X_i\}_{i\in[n]}$ is a collection of independent random variables such that $0\leq X_i\leq 1$ and $\E[X_i]=\mu_i$ for $i=1,\hdots,n$. Let $\mu=\mu_1+\cdots+\mu_n$. Then, 
\begin{align*}
\Pr(S_n\geq \mu+t)&\leq \exp\Big(-\frac{t^2}{2(\mu+t/3)}\Big) \mbox{ for any } t>0,\\
\Pr(S_n\leq \mu-t)&\leq \exp\Big(-\frac{t^2}{2(\mu-t/3)}\Big) \mbox{ for any } t\leq \mu.
\end{align*}
Further, for any $c>1$, 
\[\Pr(S_n\geq c\mu)\leq \exp\big(-\mu(c\log(c/\e)+1)\big).\]
\end{lemma}

As a preliminary application of Lemma~\ref{lem:chernoff}, we prove the following simple fact for the number of edges of a small-world graph $G\sim \Gc_{n,r}$ (recall that $N=(2n+1)^2$ is the number of vertices in $G$).
\begin{lemma}\label{lem:numedges}
Let $r\geq 0$. Then, with probability $1-\exp(-\Omega(n))$ over the choice of $G=(V,E)\sim \Gc_{n,r}$, it holds that $|E|=5N/2+\Theta(N^{3/4})$.
\end{lemma}
\begin{proof}
Let $\binom{V}{2}:=\big\{\{u,w\}\mid u,w\in V, u\neq w\big\}$ denote the set of all unordered pairs of vertices. For $\{u,w\}\in \binom{V}{2}$, let $Y_{u,w}$ be the indicator r.v. that there is a long-range edge between $u,w$. Moreover, let
\[Y=\sum_{\{u,w\}\in \binom{V}{2}}Y_{u,w}, \mbox{ and note that } |E|=2N+Y.\]
By the definition of the model $\Gc_{n,r}$, every vertex adds in expectation one long-range edge and hence $\E_{\Gc_{n,r}}[Y]=N/2$.
The lemma follows by applying Lemma~\ref{lem:chernoff}.
\end{proof}

\subsection{Edge isoperimetric inequality on the torus}
We will use the following edge isoperimetric inequality on the torus to lower bound the conductance of small sets $S$ (i.e., sets $S$ with $|S|=O(\log n)$). 

\begin{theorem}[{\cite{BLtorus}}]\label{thm:torusexpansion}
Let $n$ be a positive integer and consider the torus $T=(B_n,E_n)$. For every nonempty set $S\subseteq B_n$ with $|S|\leq N/2$, it holds that 
\[|\partial^* S|\geq \min\{2n+1,2\sqrt{|S|}\}.\]
\end{theorem}

\section{Partitioning the torus, box-like sets and  the box-core}\label{sec:part}
In this section, we partition the torus appropriately and formalise the notion of box-like sets. We also introduce the box-core of a set which will be crucial to do the refined union bound described in Section~\ref{sec:outline}.  Note that all these notions are with respect to the torus (i.e., the random small-world graph is irrelevant in this section).

Namely, to capture the trade-off between the low edge-expansion of sets which are union of boxes and the high edge-expansion of spread-out sets, we partition the torus on square boxes of small side length $\ell$ and study how the set $S$ intersects these boxes. In the case where $\ell$ divides $2n+1$ (the side length of the torus), we can obviously choose all the boxes to be squares.  To handle integrality issues, we will allow the boxes to be rectangles with only slightly unbalanced sides.

For  integers $\ell_1,\ell_2\geq 1$ and integers $m_1,m_2$, we will refer to the set 
\[\{m_1,\hdots,m_1+\ell_1-1\}\times \{m_2,\hdots,m_2+\ell_2-1\}\] 
as a  box with side lengths $\ell_1,\ell_2$.
\begin{definition}\label{def:lpartition}
Let $n,\ell$ be positive integers with $n\geq \ell$.  An $\ell$-partition  of the vertex set $B_n=\{-n,\hdots,n\}^2$ of the torus is a partition of $B_n$ into boxes with side lengths $\ell_1,\ell_2$ which satisfy
\begin{equation}\label{eq:ell1ell2}
\ell\leq \ell_1,\ell_2\leq 2\ell.
\end{equation} 
Note that the box sides $\ell_1,\ell_2$ need not be the same for every box in the partition, but  they must satisfy \eqref{eq:ell1ell2} for every box (and hence such a partition exists by the natural construction\footnote{Namely, starting from the upper left corner of $B_n$, use $\ell\times \ell$ square boxes to cover the largest possible (contiguous) square of $B_n$. Then, enlarge the end boxes of this square to go all the way to the boundaries of $B_n$, this increases only one dimension of these boxes by a factor of at most 2; it now only remains to cover a square at the lower right corner whose side length is at most $2\ell$.}).
\end{definition}

For the rest of this section, we will fix $n,\ell$ and an arbitrary $\ell$-partition of the vertex set of the torus, which we denote by $\mathcal{U}=\{U_{1},\hdots, U_Q\}$ (i.e., $\bigcup_{i\in[Q]} U_i=B_n$ and the sets $U_{i},\, i\in[Q]$ are pairwise disjoint boxes with side lengths satisfying \eqref{eq:ell1ell2}). 
\begin{definition}\label{def:boxlike}
We say that the set $S\subseteq B_n=\{-n,\hdots,n\}^2$ is \emph{box-like}  if for every $U\in \mathcal{U}$ it holds that 
\begin{equation}\label{eq:boxlikecon}
\mbox{ either } S\cap U=\emptyset, \mbox{ or } S\cap U=U.
\end{equation}
\end{definition}
Thus, a box-like set is a union of boxes of the $\ell$-partition of the torus.

\begin{definition}\label{def:boxcore}
The \emph{box-core} of a set $S\subseteq B_n=\{-n,\hdots,n\}^2$, denoted by $\boxcore(S)$, is the largest subset of $S$ which is box-like.
\end{definition}
Note that the box-core of a set $S$ can be found simply by just going over the boxes $U$ in the $\ell$-partition and checking which boxes satisfy \eqref{eq:boxlikecon}. With these definitions, we are now ready to state the main edge-expansion property of the torus which we are going to utilise for $r\leq 2$.
\begin{lemma}\label{lem:torusdecom2}
Let $n,\ell\geq 1$ be integers with $n\geq \ell$. The following holds for any $\ell$-partition of the torus $T=(B_n,E_n)$ and any $\eta\in (0,1)$.

Every set $S\subseteq B_n$ satisfies either that $\big|\boxcore(S)\big|\geq (1-\eta)|S|$,  or else, $|\partial^* S|\geq \eta|S|/(4\ell^2)$.
\end{lemma}
\begin{proof}
 We will use $G_1=(U_1,F_1),\hdots, G_Q=(U_Q,F_Q)$ to denote the subgraphs of the torus induced on the vertex sets $U_1,\hdots,U_Q$, respectively. Note that, for all $i\in[Q]$, for any set $S_i\subseteq U_i$ with $|S_i|<|U_i|$, it holds that 
\begin{equation}\label{eq:SiFi2}
|F_i\cap \partial^* S_i |\geq |S_i|/(4\ell^2),
\end{equation}
since for every $S_i\neq \emptyset$ there is at least one edge in $F_i\cap \partial^* S_i$ (if $S_i=\emptyset$, then \eqref{eq:SiFi2} holds trivially).

Now, consider an arbitrary set $S\subseteq B_n$, and set $S'=\boxcore(S)$ for convenience. Suppose that $|S'|< (1-\eta)|S|$, we will show that $|\partial^* S|\geq \eta|S|/(4\ell^2)$. Let
\[\mathcal{I}:=\big\{i\in [Q]\,\big|\, |S\cap U_i|< |U_i|\big\}.\]
By the assumption $|S'|< (1-\eta)|S|$, we have that 
\begin{equation}\label{eq:34ffg4g2}
\sum_{i\in \mathcal{I}} |S\cap U_i|\geq \eta|S|.
\end{equation}
Applying \eqref{eq:SiFi2} for $i\in \mathcal{I}$ with $S_i=S\cap U_i$, we obtain
\begin{equation}\label{eq:partialSFi2}
|F_i\cap \partial^* S|= |F_i\cap \partial^* S_i|\geq |S\cap U_i|/(4\ell^2).
\end{equation}
Combining \eqref{eq:34ffg4g2} and \eqref{eq:partialSFi2}, we obtain
\[|\partial^* S|\geq \sum_{i\in \mathcal{I}}|\partial^* S\cap F_i|\geq \sum_{i\in \mathcal{I}} |S\cap U_i|/(4\ell^2)\geq \eta|S|/(4\ell^2),\]
as wanted. This concludes the proof of the lemma.
\end{proof}

\section{The number of connected sets in the small-world network}\label{sec:connectedsets}

Addario-Berry and Lei \cite{AL} analysed the mixing time of the lazy random walk on the Newmann-Watts random graph by using a technique of Fountoulakis and Reed~\cite{FR} --- analysing the conductance of connected sets. They bounded the number of connected sets by examining a related Galton-Watson tree. We use a similar approach but the details are more complicated because the connection probabilities differ between different nodes. The main lemma we will prove in this section is the following.

\newcommand{\statelemnumcon}{Let $r\geq 0$ and $n,\ell$ be positive integers with $n\geq \ell$. Let $\mathcal{U}=\{U_{1},\hdots, U_Q\}$ be an arbitrary $\ell$-partition of the torus $T=(B_n,E_n)$.

For $G=(V,E)\sim \Gc_{n,r}$, let $W_q$ denote the number of box-like sets $S$ which are connected in $G$ and are the union of exactly $q$ boxes in $\mathcal{U}$. Then, $\E_{\Gc_{n,r}}[W_q]\leq n^2 (40\ell^2)^{q}$.
}
\begin{lemma}\label{lem:numcon}
\statelemnumcon
\end{lemma} 

\subsection{The number of subtrees of a Galton-Watson Tree with given size}

Let $X$ be a random variable supported on the non-negative integers. Recall that the Galton-Watson branching process with offspring distribution $X$ is a random rooted tree $\mathcal{T}$ where for every vertex of the tree the number of offspring vertices is independent with distribution $X$. For integer $j\geq 0$, let $p_j=\Pr[X=j]$ and 
\[q_j=\sum_{i\geq j}p_i \frac{i!}{(i-j)!}.\]
Note that $q_j$ is the expected number of ways to choose and order exactly $j$ children of the root. We will refer to the sequence $\mathbf{q}:=\{q_j\}_{j\geq 0}$ as the \emph{profile} of the random variable $X$.

\begin{lemma}[{\cite[Lemma 2]{AL}}]\label{lem:galton}
Let $\mathcal{T}$ be a Galton-Watson process with offspring distribution $X$. Suppose that there exists $C>0$ such that the profile $\mathbf{q}=\{q_j\}_{j\geq 0}$ of $X$ satisfies $q_j\leq C^j$ for all integer $j\geq 0$. 

Then, for all integer $k\geq 1$, the expected number of subtrees of $\mathcal{T}$ which contain the root and have exactly $k$ vertices is at most $(4C)^{k-1}$.
\end{lemma}

For us, the relevant offspring distribution will be a sum of independent binomial random variables. The following lemmas are implicit in \cite{AL}, we prove them for completeness.
\begin{lemma}\label{lem:binqj}
Let $n$ be a positive integer and $0< p\leq 1$. Let $Y=\mathrm{Bin}(n,p)$. Then, the profile $\mathbf{q}=\{q_j\}_{j\geq 0}$ of $Y$ satisfies $q_j\leq (n p)^{j}$. 
\end{lemma}
\begin{proof}[Proof of Lemma~\ref{lem:binqj}]
Let $j\geq 0$ be an arbitrary integer. We may assume that $j\leq n$, otherwise $q_j=0$ and the inequality $q_j\leq (n p)^{j}$ holds trivially. We have
\begin{align*}
q_j&=\sum^{n}_{i= j}\binom{n}{i}p^i(1-p)^{n-i}\frac{i!}{(i-j)!}=\sum^{n}_{i= j}\frac{n!}{(n-i)!(i-j)!}p^i(1-p)^{n-i}\\
&\leq n^j p^j\sum^{n}_{i= j}\frac{(n-j)!}{(n-i)!(i-j)!}p^{i-j}(1-p)^{n-i}=(np)^j.
\end{align*}
This completes the proof.
\end{proof}
\begin{lemma}\label{lem:sumindqj}
Let $Y^{(1)}, Y^{(2)},\hdots, Y^{(n)}$ be independent random variables supported on the non-negative integers, and let $Y=\sum^n_{t=1}Y^{(t)}$. Let $\mathbf{q}=\{q_j\}_{j\geq 0}$ denote the profile of $Y$ and, for $t=1,\hdots,n$, let $\mathbf{q}^{(t)}=\{q_j^{(t)}\}_{j\geq 0}$ denote the profile of $Y^{(t)}$.

Then, for all integer $j\geq 0$, it holds that 
\[q_{j}=\sum_{\substack{j_1,\hdots,j_n\geq 0;\\ j_1+\cdots+j_n=j}} \binom{j}{j_1,\hdots,j_n}\prod^n_{t=1}q^{(t)}_{j_t}.\]
Further, assuming that $C_t$, $t=1,\hdots,n$, are such that $q^{(t)}_j\leq (C_t)^{j}$ for all nonnegative integers $j$, then it holds that $q_j\leq C^{j}$ for all nonnegative integers $j$, where $C=C_1+\cdots+C_t$.
\end{lemma}
\begin{proof}[Proof of Lemma~\ref{lem:sumindqj}]
The first part follows probabilistically, using the   independence of the random variables  $Y^{(1)},\hdots,Y^{(n)}$ and the combinatorial interpretation of the profile. The second part of the lemma follows immediately from the first part using the multinomial theorem.
\end{proof}
Combining Lemmas~\ref{lem:galton}, \ref{lem:binqj} and~\ref{lem:sumindqj}, we obtain the following corollary.
\begin{corollary}\label{cor:bintrees}
Let $\mathcal{T}$ be a Galton-Watson process with offspring distribution $X$, where $X$ is a sum of independent binomial random variables. 

Then, for all integer $k\geq 1$, the expected number of subtrees of $\mathcal{T}$ which contain the root and have exactly $k$ vertices is at most $(4E[X])^{k-1}$.
\end{corollary}

\subsection{Dominating the number of trees in random graphs by a branching process}
To use Corollary~\ref{cor:bintrees} for the proof of Lemma~\ref{lem:numcon}, we need to bound the number of connected (box-like) sets in $G\sim\Gc_{n,r}$ by the number of subtrees in an appropriately defined Galton-Watson tree. We do this in a rather general setup so that we can account for the graph distribution $\Gc_{n,r}$ and its relevant variants that we will need in the proof of Lemma~\ref{lem:numcon}. Our goal is to account for the fact that the graph is obtained in a non-uniform way in the sense that each pair of vertices are connected with probability that depends on the labelling of the pair.  

Throughout this section, we will let $n$ be a positive integer,  $[n]$ be the set $\{1,\hdots,n\}$ (the vertex set of the graph)  and $\pb=\{p_{ij}\}_{i,j\in [n]}$ be a symmetric matrix whose entries are in the interval $[0,1]$. 

\begin{definition}[The graph distribution $\Gc_{n,\pb}$]\label{def:Gnp}
The random graph $G=([n],E)\sim \Gc_{n,\pb}$ is obtained  by adding independently, for every pair of vertices $i,j\in [n]$,  the edge $\{i,j\}$ with probability $p_{ij}$.
\end{definition}

\begin{definition}[The tree process $\Tc^{\,i}_{n,\pb}$]\label{def:Tnp}
Let $i\in[n]$. The tree process $\Tc^{\,i}_{n,\pb}$ is a random tree rooted at $i$ whose nodes at distance $\ell\geq 1$ from the root are labelled with an element of $[n]^{\ell+1}$ (that is, a node at distance $\ell$ from the root will be labelled with a sequence of $\ell+1$ elements of $[n]$).

Initialise the process by setting $R_0=\{i\}$. For $\ell\geq 0$, suppose that we have constructed $R_\ell\subseteq [n]^{\ell+1}$. To construct $R_{\ell+1}$, for each node $\xb=(x_0,\hdots,x_\ell)\in R_\ell$, do the following:
\begin{equation}\label{eq:4r34}
\begin{aligned}
&\mbox{For each $j\in [n]$, toss independently a random coin with heads probability $p_{x_\ell,j}$.  If the coin }\\
&\mbox{comes up heads, then add $\yb:=(x_0,\hdots,x_\ell,j)$ in $R_{\ell+1}$ and connect $\xb$ and $\yb$ with an edge.}
\end{aligned}
\end{equation}
\end{definition}
Note that in \eqref{eq:4r34} the children of a node $\xb=(x_0,\hdots,x_\ell)$ are added with probabilities which depend only on $x_\ell$. We are now ready to show the following.
\begin{lemma}\label{lem:conngalton}
Let $i\in [n]$ and $k\geq 1$ be an integer. For $G\sim \Gc_{n,\pb}$, let $W_k(i)$ be the number of sets $S\subseteq[n]$ which are connected in $G$ and satisfy $|S|=k$ and $i\in S$. Also, for the tree process $\Tc^{\,i}_{n,\pb}$, let  $W_k'(i)$ be the number of subtrees with $k$ vertices  containing the root $i$. Then,
\[\E_{\Gc_{n,\pb}}[W_k(i)]\leq \E_{\Tc^{\,i}_{n,\pb}}[W_k'(i)].\]
\end{lemma}
\begin{proof}
Denote by $\Tb_k^{i}$  the set of all  labelled trees on the vertex set $[n]$  with exactly $k$ vertices which include the vertex $i$. For a tree $\tb\in \Tb_{k}^{i}$, let $1_\tb$ be the indicator r.v. that $\tb$ is a subgraph of $G$. Then, we have that $W_k(i)\leq \sum_{\tb\in \Tb_{k}^{i}}1_\tb$ and therefore
\[\E_{\Gc_{n,\pb}}[W_k(i)]\leq \sum_{\tb\in \Tb_{k}^{i}}\E_{\Gc_{n,\pb}}[1_\tb].\]

Let $\tb$ be a tree in $\Tb_{k}^{i}$. There is a natural way to map $\tb$ to an outcome of the tree process $\mathcal{T}^{\,i}_{n,\pb}$. In particular, let $M(\tb)$ be the tree which is isomorphic to $\tb$, where a node originally labelled $v$ in $\tb$ is relabelled by the path starting at the root $i$ and ending in $v$ (we view the path  as an ordered tuple). Let $1_{M(\tb)}$ be the indicator r.v. that $M(\tb)$ is a subgraph of $\mathcal{T}^{\,i}_{n,\pb}$. Then, we have $W_k'(i)\geq \sum_{\tb\in \Tb_{k}^{i}}1_{M(\tb)}$ and therefore
\[\E_{\Tc^{\,i}_{n,\pb}}[W_k'(i)]\geq \sum_{\tb\in \Tb_{k}^{i}}\E_{\Tc^{\,i}_{n,\pb}}[1_{M(\tb)}].\]
All that remains to observe is that, for an arbitrary tree $\tb\in \Tb^{i}_k$, it holds that 
\[\E_{\Gc_{n,\pb}}[1_\tb]=\E_{\Tc^{\,i}_{n,\pb}}[1_{M(\tb)}].\qedhere\]
\end{proof}

\subsection{Proof of  Lemma~\ref{lem:numcon}}
We are now ready to prove Lemma~\ref{lem:numcon}.
\begin{lemnumcon}
\statelemnumcon
\end{lemnumcon}
\begin{proof}[Proof of Lemma~\ref{lem:numcon}]
Let $G=(V,E)\sim \Gc_{n,r}$. Consider the box graph $G_\mathcal{U}=(\mathcal{U}, E_\mathcal{U})$ induced by the $\ell$-partition of the torus $\mathcal{U}=\{U_{1},\hdots, U_Q\}$, where two boxes  $U_i$ and $U_j$ are connected by an edge if there is an edge in $G$ between $U_i$ and $U_j$. Note that a box-like set $S$ which is connected in $G$ and is the union of $q$ boxes in $\mathcal{U}$ corresponds to a set $S'\subseteq \mathcal{U}$ with $|S'|=q$ which is connected in $G_\mathcal{U}$.

For a box $U_i\in \mathcal{U}$, let $W_{q}(i)$ denote the number of  sets $S'\subseteq \mathcal{U}$ with $|S'|=q$ which are connected in $G_\mathcal{U}$  such that $U_i\in S'$. Clearly, $W_q\leq \sum_{i\in [Q]}W_{q}(i)$ and hence, by linearity of expectation, to prove the lemma it suffices to show that $\E_{\Gc_{n,r}}[W_{q}(i)]\leq (40\ell^2)^{q-1}$ for all $i\in [Q]$ (using the crude bound $|\mathcal{U}|=Q\leq 10n^2$).

To bound $W_q(i)$, we will use Lemma~\ref{lem:conngalton}. In particular, for two boxes $U_j$ and $U_k$, let $p_{jk}$ be the probability that $U_j$ and $U_k$ are connected in the box graph $G_\mathcal{U}$ and let $\pb$ denote the matrix $\{p_{jk}\}_{j,k\in[Q]}$. Note that $G_\mathcal{U}$ follows the graph distribution $\Gc_{Q,\pb}$ (cf. Definition~\ref{def:Gnp}), so by Lemma~\ref{lem:conngalton} we have that 
\[\E_{\Gc_{Q,\pb}}[W_q(i)]\leq \E_{\mathcal{T}^{i}_{Q,\pb}}[W_q'(i)].\]
where $W_q'(i)$ denotes the number of subtrees with $q$ vertices containing the root $i$ in the tree process $\mathcal{T}^i_{Q,\pb}$ (cf. Definition~\ref{def:Tnp}).

Let $\Tc$ be a Galton-Watson tree with offspring distribution $Y:=4+\sum^{4\ell^2}_{j=1}X_j$, where $X_j$ are i.i.d. random variables distributed as
\[X=\sum^n_{\ell'=2}\mathrm{Bin}\big(4\ell',(\ell')^{-r}/Z\big)+\sum^{2n}_{\ell'=n+1}\mathrm{Bin}\big(4(2n+1-\ell'),(\ell')^{-r}/Z\big).\]
Note that $X$ has the same distribution as the number of long-range neighbours of an arbitrary vertex in $G$ and hence $\E[X]=1$ by the definition of the model $\Gc_{n,r}$. It follows that $\E[Y]=4+4\ell^2$ and hence by Corollary~\ref{cor:bintrees}, we obtain that the expected number of subtrees of $\Tc$ with $q$ vertices containing the root is at most $(4\E[Y])^{q-1}\leq (40\ell^2)^{q-1}$.

Therefore, to prove the lemma it suffices to couple the tree process $\Tc^{\,i}_{Q,\pb}$ with the Galton-Watson tree $\Tc$ so that $\Tc^{\,i}_{Q,\pb}$ is a subtree of $\Tc$ (when we view them as unlabelled graphs). Consider an arbitrary node in $\Tc^{\,i}_{Q,\pb}$; this corresponds to a tuple $(U_i,\hdots,U)$ for some box $U\in \mathcal{U}$ and therefore the number of the node's children in step \eqref{eq:4r34} is distributed as  the  number of  neighbours of $U$ in the box graph $G_\mathcal{U}$.  For a vertex $u\in U$, let $X_u$ be the number of long-range neighbours of $u$ in $V\backslash U$, so that the number of neighbours of the box $U$ is dominated above by $4+\sum_{u\in U} X_u$. Note that the variables $\{X_u\}_{u\in U}$ are independent and each $X_u$ is dominated above by the random variable $X$. It follows that the number of neighbours of an arbitrary box $U$ in the box graph $G_\mathcal{U}$ is dominated above by $Y$ and therefore the number of children of an arbitrary node in the tree process $\Tc^{\,i}_{Q,\pb}$  is also dominated above by $Y$. Hence, by revealing the processes $\Tc^{\,i}_{Q,\pb}$ and $\Tc$ in a breadth-first search manner, we can couple them so that $\Tc^{\,i}_{Q,\pb}$ is a subtree of $\Tc$.

This concludes the proof of Lemma~\ref{lem:numcon}.
\end{proof}

\section{Upper bound on the mixing time for $r<2$}\label{sec:rl2}

In this section we prove the upper bound $O(\log n)$ on the mixing time for the small-world network model $\Gc_{n,r}$ when $r<2$.

\subsection{Proof Outline}
In this section, for  $G=(V,E)\sim \Gc_{n,r}$, we will be interested in the expansion properties of a  set $S\subseteq V$; one particular quantity of interest will be the size of $\partial S$, i.e.,  the number of edges with exactly one endpoint in $S$, which will allow us to bound the conductance $\Phi(S)$.  However, we will need slightly more information in our later arguments, which is captured by the following definition. 
\begin{definition}\label{def:expanding}
Let $G=(V,E)\sim \Gc_{n,r}$ and $\epsilon,c>0$ be arbitrary real numbers. A set $S\subseteq V$ is called $(\epsilon,c)$-expanding in $G$ if for all subsets $S'\subseteq S$ with $|S'|\geq (1-\epsilon)|S|$, it holds that 
\[|\partial S\cap \partial S'|\geq c |S'|,\]
i.e., there are at least $c |S'|$ edges with one endpoint in $S'$ and one endpoint in $V\backslash S$.
\end{definition}
Note that if a set $S$ is $(\epsilon,c)$-expanding for some $\epsilon,c>0$, then trivially $|\partial S|\geq c |S|$. Intuitively, the fact that $S$ is $(\epsilon,c)$-expanding captures that these $|\partial S|$ edges are ``well-distributed" within the set $S$.

\newcommand{\statelemarbexp}{Let $r\in[0,2)$. Then, there exist constants $\epsilon,c>0$  such that, for all sufficiently large integers $n$,   for all sets $S\subseteq V$ with $|S|\leq 99 N/100$, it holds that 
\[\mbox{$\Pr_{\Gc_{n,r}}$}\big(\mbox{$S$ is $(\epsilon,c)$-expanding}\big)\geq 1-\exp(-c|S|).\]
}
\begin{lemma}\label{lem:arbexp}
\statelemarbexp
\end{lemma}
The proof of Lemma~\ref{lem:arbexp} is given in Section~\ref{sec:lowboundexp}. 
To utilise Lemma~\ref{lem:arbexp}, we will need the following simple observation.
\begin{lemma}\label{lem:condSconn}
Let $r\geq 0$ and $n$ be a positive integer. Then, for all $\epsilon, c>0$,  for all sets $S\subseteq V$, it holds that 
\[\mbox{$\Pr_{\Gc_{n,r}}$}\big(\mbox{$S$ is $(\epsilon,c)$-expanding}\, \big|\, S\mbox{ is connected}\big)=\mbox{$\Pr_{\Gc_{n,r}}$}\big(\mbox{$S$ is $(\epsilon,c)$-expanding}\big).\]
\end{lemma}
\begin{proof}
For $u,w\in V$, let $Y_{u,w}$ be the indicator r.v. that there is a long-range edge between $u,w$. Observe that the event that $S$ is $(\epsilon,c)$-expanding is completely determined by the random variables $\{Y_{u,w}\}_{u\in S,w\in V\backslash S}$, while the event that $S$ is connected  is completely determined by the random variables $\{Y_{u,w}\}_{u\in S,w\in S}$. It follows that the two events are independent.
\end{proof}

Recall the definitions of an $\ell$-partition of the torus and box-like sets (cf. Section~\ref{sec:part}). We show the following lemma.
\begin{lemma}\label{lem:bigexpbigexp}
Let $r\in [0,2)$. There exist constants $\epsilon, c, \ell_0>0$ such that for all sufficiently large integers $n$ and every integer $\ell\in [\ell_0,n]$, the following holds for any $\ell$-partition of the torus $T=(B_n,E_n)$, with probability $1-O(1/n^2)$ over the choice of $G=(V,E)\sim \Gc_{n,r}$.

For every box-like set $S$ which is connected in $G$ and satisfies $100\ell^2\log n\leq |S|\leq 99N/100$, it holds that $S$ is $(\epsilon,c)$-expanding.
\end{lemma}
\begin{proof}
By Lemmas~\ref{lem:arbexp} and~\ref{lem:condSconn}, there exist constants $\epsilon,c>0$ such that all sets $S\subseteq V$ with $|S|\leq 99N/100$ satisfy
\begin{equation}\label{eq:rtgg}
\mbox{$\Pr_{\Gc_{n,r}}$}\big(\mbox{$S$ is  not  $(\epsilon,c)$-expanding}\, \big|\, S\mbox{ is connected}\big)\leq \exp(-c|S|).
\end{equation}
Let $\ell_0$ be a constant such that for all $\ell\geq \ell_0$ it holds that  $\exp(-c\ell^2)(40 \ell^2)\leq 1/\e$. For $\ell\geq \ell_0$, consider an arbitrary $\ell$-partition $\mathcal{U}=\{U_1,\hdots, U_Q\}$ of the torus and note that $Q\leq N/\ell^2$ (since every box  $U\in\mathcal{U}$ contains at least $\ell^2$ vertices).

For a box-like set $S$, denote by $q_S$ the number of boxes $U\in\mathcal{U}$ such that $S\cap U\neq \emptyset$. Since $S$ is box-like, for every $U\in\mathcal{U}$ such that $S\cap U\neq \emptyset$  we have that $|S\cap U|=|U|$ and hence $q_S\ell^2\leq |S|\leq 4q_S\ell^2$. The assumption that $100\ell^2\log n\leq |S|\leq 99N/100$ therefore translates into $25\log n\leq q_S\leq 99N/(100\ell^2)$.

Let $G=(V,E)\sim \Gc_{n,r}$. For an integer $q\geq 1$, let $\mathcal{E}_q$ be the event that there exists a box-like set $S$ with $q_S=q$ which is connected in $G$ and which is \emph{not} $(\eps,c)$-expanding. To prove the lemma, it therefore suffices to show that 
\begin{equation}\label{eq:ggreverv3v4}
\mbox{$\Pr_{\Gc_{n,r}}$}\bigg(\bigcup_{25\log n\leq q\leq 99N/(100\ell^2)} \mathcal{E}_q\bigg)\leq 1/n^2.
\end{equation}

By Lemma~\ref{lem:numcon}, we have that the number $W_q$ of box-like sets $S$ with $q_S=q$ that are connected in $G$ satisfies the bound
\begin{equation}\label{eq:wqell}
\E_{\Gc_{n,r}}[W_q]\leq n^2 (40\ell^2)^{q}.
\end{equation}
Consider an arbitrary integer $q$ such that $25\log n\leq q\leq 99N/(100\ell^2)$ and let $\mathcal{F}_q$ denote the set of box-like sets $S$ with $q_S=q$. Then, using \eqref{eq:rtgg}~and~\eqref{eq:wqell}, we have
\begin{align*}
\mbox{$\Pr_{\Gc_{n,r}}$}(\mathcal{E}_q)&\leq \sum_{S\in \mathcal{F}_q} \mbox{$\Pr_{\Gc_{n,r}}$}(\mbox{$S$ is connected})\, \mbox{$\Pr_{\Gc_{n,r}}$}\big(\mbox{$S$ is not $(\epsilon,c)$-expanding}\, \big|\, S\mbox{ is connected}\big)\\
&\leq \exp(-cq\ell^2)\sum_{S\in \mathcal{F}_q} \mbox{$\Pr_{\Gc_{n,r}}$}(\mbox{$S$ is connected})=\exp(-cq\ell^2)\, \E_{\Gc_{n,r}}[W_q]\\
&\leq  n^2 \exp(-cq\ell^2)(40\ell^2)^{q}\leq n^{-20},
\end{align*}
where the last inequality follows by the choice of $\ell$ and $\ell_0$. By a union bound over the possible values of $q$, we therefore obtain \eqref{eq:ggreverv3v4}, as wanted.
\end{proof}

Using these lemmas, we obtain the following conductance bounds (recall that $\pi$ denotes the stationary distribution of the lazy random walk, cf. Section~\ref{sec:lazyrandomwalk}).

\newcommand{\statelemconrltwo}{Let $r\in [0,2)$. There exist constants $\rho,\tau,\chi>0$ such that the following hold for all sufficiently large integers  $n$ with probability $1-O(1/n^2)$ over the choice of the graph $G\sim\Gc_{n,r}$. 

For every (nonempty) connected set $S$ in $G$ with $\pi(S)\leq 1/2$, the conductance $\Phi(S)$ satisfies 
\begin{equation*}
\Phi(S)\geq \begin{cases} \tau&\mbox{ if } |S|\geq \rho\log n\\ \tau/(\sqrt{|S|})&\mbox{ if } |S|\leq \rho\log n.\end{cases}
\end{equation*}
Further, for every connected set $S$ with $\Phi(S)<\tau$ and $\pi(S)\leq 1/2$, it holds that $\pi(S)\leq \chi |S|/N$.
}
\begin{lemma}\label{lem:conrltwo}
\statelemconrltwo
\end{lemma}
The proof of Lemma~\ref{lem:conrltwo} is given in Section~\ref{sec:cdewds}.

\begin{corollary}\label{cor:rl2}
Let $r\in [0,2)$ and $n$ be a sufficiently large integer. With probability $1-O(1/n^2)$ over the choice of
the graph $G\sim\Gc_{n,r}$, the lazy random walk on $G$
 satisfies $\Tmix=O(\log n)$.
\end{corollary}
\begin{proof}
Let $G=(V,E)\sim\Gc_{n,r}$. By a union bound, we have that with probability $1-O(1/n^2)$ the graph $G$ satisfies Lemmas~\ref{lem:numedges} and~\ref{lem:conrltwo}. For all such graphs $G$, we will show that  $\Tmix=O(\log n)$.

Recall from Theorem~\ref{thm:FR} that there exists an absolute constant $C>0$ such that 
\begin{equation*}\tag{\ref{eq:phij12}}
\Tmix\leq C \sum^{\left\lceil \log_2(\pi_{\mathrm{min}}^{-1})\right\rceil}_{j=1}\bigg(\frac{1}{\widetilde{\Phi}(2^{-j})}\bigg)^2.
\end{equation*}
where $\widetilde{\Phi}(p):=\min\{\Phi(S) \mid S\mbox{ is connected},\, p/2\leq \pi(S)\leq p\}$, and $\widetilde{\Phi}(p)=1$ if the minimization is over an empty set.  By Lemma~\ref{lem:numedges}, we have that $|E|\leq 3N$ and hence $\pi_{\mathrm{min}}=\min_{v\in V} \{d_v/(2|E|)\}\geq 4/(2|E|)\geq 2/(3N)$, so that $\log_2(1/\pi_{\mathrm{min}})=O(\log n)$.

Let $\mathcal{J}$ be the set of indices $j$ in \eqref{eq:phij12} such that $\widetilde{\Phi}(2^{-j})< \tau$, where $\tau$ is the constant in Lemma~\ref{lem:conrltwo}. The contribution to the sum in \eqref{eq:phij12} from indices $j\notin \mathcal{J}$ is clearly at most $O(\log n)$, so we only need to focus on the contribution  from indices $j\in \mathcal{J}$.

For $j\in \mathcal{J}$, there exists a connected set $S$ such that $2^{-j-1}\leq \pi(S)\leq 2^{-j}$ and $\Phi(S)\leq \tau$, so by the first part of Lemma~\ref{lem:conrltwo}, we have $|S|\leq \rho \log n$.  Moreover, by the second part of Lemma~\ref{lem:conrltwo}, we have that $\pi(S)\leq \chi |S|/N\leq \rho \chi \log n /N$, so we have 
\begin{equation}\label{eq:ef332}
2^{-j-1}\leq \rho \chi \log n /N\mbox{ for all indices $j\in \mathcal{J}$.}
\end{equation}

Now, for $j\in \mathcal{J}$, consider an arbitrary connected set $S$  satisfying $2^{-j-1}\leq \pi(S)\leq 2^{-j}$. Since $d_v\geq 4$ for all $v\in V$ and $|E|\leq 3N$ (by Lemma~\ref{lem:numedges}), we have 
\[\pi(S)=\frac{1}{2|E|}\sum_{v\in S}d_v\geq\frac{2|S|}{3N}\]
We obtain that $|S|\leq N/2^{j-1}$ and therefore, by Lemma~\ref{lem:conrltwo}, we have that $\Phi(S)\geq \tau (2^{j-1}/N)^{1/2}$. Since $S$ was an arbitrary set satisfying $2^{-j-1}\leq \pi(S)\leq 2^{-j}$, it follows that $\widetilde{\Phi}(2^{-j})\geq \tau (2^{j-1}/N)^{1/2}$, i.e.,
\begin{equation}\label{eq:ef333}
\bigg(\frac{1}{\widetilde{\Phi}(2^{-j})}\bigg)^2\leq \frac{N}{\tau^2\, 2^{j-1}}\mbox{ for all indices $j\in \mathcal{J}$.}
\end{equation}
From \eqref{eq:ef333}, we obtain that the contribution to the sum in \eqref{eq:phij12} from indices $j\in \mathcal{J}$ is bounded by a geometric series whose largest term is bounded by $O(\log n)$ from  \eqref{eq:ef332}. Therefore, the total contribution is bounded by  $O(\log n)$.

This yields that $\Tmix=O(\log n)$.
\end{proof}

\subsection{Lower bounding the expansion -- Proof of Lemma~\ref{lem:arbexp}}\label{sec:lowboundexp}
In this section, we prove Lemma~\ref{lem:arbexp} which we restate here for convenience.

\begin{lemarbexp}
\statelemarbexp
\end{lemarbexp}
\begin{proof}[Proof of Lemma~\ref{lem:arbexp}]
Let $\eta:=1/100$ and consider an arbitrary set $S\subseteq V$ with $|S|\leq (1-\eta) N$. 

Let $S'\subseteq S$ be a subset of $S$. For $u,w\in V$, let $Y_{u,w}$ be the indicator r.v. that there is an edge from $u$ to $w$ in $G$. Note that 
\begin{equation}\label{lem:partXpl}
|\partial S\cap \partial S'|\geq Y(S'), \mbox{ where } Y(S'):=\sum_{u\in S',w\in V\backslash S}Y_{u,w}.
\end{equation}
Note that $\{Y_{u,w}\}_{u\in S',w\in V\backslash S}$ is a collection of  independent random variables. Therefore, for $\mu_{S'}:=\E_{\Gc_{n,r}}[Y(S')]$, we have by Lemma~\ref{lem:chernoff} that 
\begin{equation}\label{lem:applychernoff}
\mbox{$\Pr_{\Gc_{n,r}}$}\big(Y(S')\leq \tfrac{1}{2}\mu_{S'}\big)\leq \exp\big(-\tfrac{1}{10}\mu_{S'}\big).
\end{equation}
Combining \eqref{lem:partXpl} and  \eqref{lem:applychernoff}, we obtain that 
\begin{equation}\label{eq:g4revrwr}
\mbox{$\Pr_{\Gc_{n,r}}$}\big(|\partial S \cap \partial S'|\leq \tfrac{1}{2}\mu_{S'}\big)\leq \exp\big(-\tfrac{1}{10}\mu_{S'}\big).
\end{equation}
We will  show that there exists a constant $\tau>0$ (which does not depend on $S$) such that for all $\epsilon\in(0,1)$, it holds that 
\begin{equation}\label{eq:toprove1312}
\mbox{$\mu_{S'}\geq \tau |S'|$ for all sets $S'\subseteq S$ with $|S'|\geq (1-\epsilon)|S|$. }
\end{equation}
We first conclude the proof of the lemma assuming \eqref{eq:toprove1312}. Let $\epsilon>0$ be a sufficiently small constant so that 
\begin{equation}\label{eq:sdf4ffv3v33}
\mbox{$\epsilon\leq \min\{1/2,\tau/100\}$ and $(\e/\epsilon)^{\epsilon}\leq \exp(\tau/100)$}.
\end{equation}
Note that such a constant exists since $\tau>0$ and $(\e/\epsilon)^{\epsilon}\downarrow 1$ as $\epsilon\downarrow 0$. Then, using \eqref{eq:g4revrwr} and \eqref{eq:toprove1312}, we obtain by a union bound over all possible choices of $S'\subseteq S$ with $|S'|\geq (1-\epsilon)|S|$ that
\begin{align}
\mbox{$\Pr_{\Gc_{n,r}}$}\big(\mbox{$S$ is not $(\epsilon,\tau/2)$-expanding}\big)&\leq \sum^{|S|}_{s=(1-\epsilon)|S|}\binom{|S|}{s}\exp\big(-\tau s/10\big)\notag\\
&\leq (\epsilon |S|+1)\max_{s\in \big[(1-\epsilon)|S|,|S|\big]}\binom{|S|}{s}\exp\big(-\tau s/10\big).\label{eq:ferf334v43v}
\end{align}
Let $k:=|S|$. Using \eqref{eq:sdf4ffv3v33}, we obtain that for $s\in [(1-\epsilon)k,k]$, it holds that $\tau s\geq \tau(1-\epsilon)k\geq \tau k/2$, 
\[\binom{k}{s}=\binom{k}{k-s}\leq \Big(\frac{\e k}{k-s}\Big)^{k-s}\leq \Big(\frac{\e}{\epsilon}\Big)^{\epsilon k}\leq\exp(\tau k/100)\ \mbox{ and }\  \epsilon k+1\leq \exp(\epsilon k)\leq \exp(\tau k/100).\] 
Using these, we obtain from \eqref{eq:ferf334v43v} that
\[\mbox{$\Pr_{\Gc_{n,r}}$}\big(\mbox{$S$ is not $(\epsilon,\tau/2)$-expanding}\big)\leq \exp\big(-\tau|S|/40\big),\]
which proves the statement of the lemma with $c=\tau/40$.

It thus remains to prove that there exists a constant $\tau>0$ such that \eqref{eq:toprove1312} holds. Recall the normalising factor $Z$ given in \eqref{eq:defZ} and that $\dist{u}{w}$ denotes the distance between $u,w$ in the torus. For all sufficiently large $n$, we claim that $\E_{\Gc_{n,r}}[Y_{u,w}]\geq \dist{u}{w}^{-r}/Z$: for non-adjacent vertices $u,w$ in the torus, the inequality holds at equality by the definition of the model $\Gc_{n,r}$; for vertices $u,w$ which are adjacent  in the torus, we have $\E_{\Gc_{n,r}}[Y_{u,w}]=1\geq 1/Z$, where the last inequality holds for all sufficiently large $n$  by Lemma~\ref{lem:factor}. We therefore have that 
\[\mu_{S'}=\sum_{u\in S'}\sum_{w\in V\backslash S}\E_{\Gc_{n,r}}[Y_{u,w}]\geq \frac{1}{Z}\sum_{u\in S'}\sum_{w\in V\backslash S}\dist{u}{w}^{-r}.\]
To lower bound the last sum, we have by Lemma~\ref{lem:factor} that there exists a constant $C>0$ such that $Z\leq C n^{2-r}\leq C N n^{-r}$. Moreover, we have the trivial bounds $\dist{u}{w}\leq 2n$ and $|V\backslash S|\geq \eta N$ (since $|S|\leq (1-\eta) N$). It follows that 
\begin{equation}\label{eq:mupS}
\mu_{S'}\geq\frac{1}{Z}\sum_{w\in V\backslash S}\sum_{u\in S'}\dist{u}{w}^{-r}\geq \sum_{u\in S'}\frac{\eta N}{Z}(2 n)^{-r}\geq (\eta 2^{-r}/C)|S'|.
\end{equation}
Clearly, the constant $\tau:=\eta 2^{-r}/C>0$ does not depend on the set $S$, thus proving \eqref{eq:toprove1312} and concluding the proof of the lemma.
\end{proof}

\subsection{Conductance lower bounds -- Proof of Lemma~\ref{lem:conrltwo}}\label{sec:cdewds}
In this section, we prove the conductance bounds stated in Lemma~\ref{lem:conrltwo}. We will need the following upper bound on the average degree of a set $S$, which ensures that the number of edges \emph{within} an arbitrary set $S$ is at most linear in $|S|$. Of course, this is immediate for linear-sized sets $S$, so most of the work is to ensure that this holds for sets $S$ with relatively small cardinality.

\newcommand{\statelemavgdeg}{Let $r\in [0,2)$. Then, there exists a constant $M>0$ such that the following holds for all sufficiently large $n$ with probability $1-O(1/n^2)$ over the choice of the graph $G=(V,E)\sim\Gc_{n,r}$. 

For every set $S\subseteq V$ in $G$, it holds that $\sum_{v\in S}d_v\leq M|S|+|\partial S|$.
}
\begin{lemma}\label{lem:avgdeg}
\statelemavgdeg
\end{lemma}
\begin{proof}
Consider an arbitrary set $S\subseteq V$. Let $\binom{S}{2}:=\big\{\{u,w\}\mid u,w\in S, u\neq w\big\}$ denote the set of all unordered pairs of vertices in $S$. For $\{u,w\}\in \binom{S}{2}$, let $Y_{u,w}$ be the indicator r.v. that there is a long-range edge between $u,w$. Then, we have that
\begin{equation}\label{eq:avgupbound}
\sum_{v\in S}d_v\leq 4|S|+2Y_S+|\partial S|, \mbox{ where } Y_S:=\sum_{\{u,w\}\in \binom{S}{2}} Y_{u,w}.
\end{equation}
Let $\mu_S:=\E_{\Gc_{n,r}}[Y_S]$.  We will show that there exists a constant $\tau>0$ (independent of $S$) such that for all sets $S\subseteq V$ it holds that 
\begin{equation}\label{eq:c3cecc}
\mu_S\leq \tau |S|\big(|S|/N\big)^{1-\frac{r}{2}}.
\end{equation}

We finish the proof of the lemma assuming for now \eqref{eq:c3cecc}. Let $M'>100$ be a constant such that $M' \tau(1-r/2)\geq 50$. Since for any set $S\subseteq V$ the random variables $\{Y_{u,w}\}_{u\in S,w\in S}$ are  independent, we obtain by Lemma~\ref{lem:chernoff}  with $c=M'\tau|S|/\mu_S$ and the inequalities $20|S|\leq M'\tau(1-r/2)|S|$, $(1-r/2)\log(\e N/|S|)\leq \log(c/\e)$ that
\begin{equation*}
\mbox{$\Pr_{\Gc_{n,r}}$}\big(Y_S\geq M'\tau|S|\big)\leq \exp\Big(-20|S|\log\big(\e N/|S|\big)\Big).
\end{equation*}
By a union bound over all sets $S\subseteq V$, we obtain that the probability that there exists a set $S$ such that $Y_S\geq M'\tau |S|$ is at most
\begin{equation*}
\sum^{N}_{s=1}\binom{N}{s}\exp\Big(-20s\log\big(\e N/s\big)\Big)\leq \sum^{N}_{s=1} \Big(\frac{\e N}{s}\Big)^{s}\exp\Big(-20s\log\big(\e N/s\big)\Big)\leq \frac{1}{N^{10}}.
\end{equation*}
Combining this with \eqref{eq:avgupbound}, we obtain the statement of the lemma with the constant $M=4+2M'\tau$.

To finish the proof, it remains to show that there exists a constant $\tau>0$ such that \eqref{eq:c3cecc} holds. 
If $|S|> N/10$, we have the trivial bound $\mu_S\leq N$ (since every vertex has on average one long-range edge), and hence we can satisfy \eqref{eq:c3cecc} by choosing $\tau$ to be any constant bigger than $\tau_0=10^{2-r/2}$.
For $|S|\leq N/10$, we have the bound 
\begin{equation}\label{eq:ZZupper312}
\mu_S\leq \frac{1}{Z}\sum_{u\in S}\sum_{w\in S; w\neq u}\dist{u}{w}^{-r},
\end{equation}
where the normalising factor $Z$ is given in \eqref{eq:defZ} and $\dist{u}{w}$ denotes the distance between $u,w$ in the torus. By Lemma~\ref{lem:factor}, we have that there exists a constant $C>0$ such that $Z\leq C n^{2-r}\leq C N n^{-r}$. To bound the  sum in \eqref{eq:ZZupper312}, for $u\in S$, let 
\[Q_u:=\sum_{w\in S; w\neq u}\dist{u}{w}^{-r}.\]
Observe that $Q_u$ can only increase if we move vertices in $S$ as close as possible to the vertex $u$. In particular, let $\ell_0:=\left\lceil 3n\sqrt{|S|/N}\right\rceil\leq n$ and consider the set of vertices 
\[U:=\{w\in V\mid  \dist{v}{w}\leq \ell_0\}.\]
Note that for $\ell\leq \ell_0\leq  n$, the  number of vertices at distance $\ell$ from $u$ is $4\ell$, and hence
\[|U|=1+4\sum^{\ell_0}_{i=1}\ell=1+2\ell_0(\ell_0+1)\geq  1+|S|.\] 
Therefore,
\begin{equation*}
Q_u\leq \sum_{w\in U; w\neq u}\dist{u}{w}^{-r}=4\sum^{\ell_0}_{\ell =1}\frac{1}{\ell^{r-1}}\leq 4\bigg(1+\int^{\ell_0+1}_{1} \frac{1}{x^{r-1}} dx\bigg)\leq \frac{4}{2-r}(\ell_0+1)^{2-r},
\end{equation*}
Since $\ell_0+1\leq 6n\sqrt{|S|/N}$, we obtain that for every $u\in S$, it holds that $Q_u\leq \tau' n^{2-r}(|S|/N)^{1-\frac{r}{2}}$ where $\tau':=\frac{4}{2-r}6^{2-r}$. Plugging this into \eqref{eq:ZZupper312}, we obtain 
that \eqref{eq:c3cecc} holds for any constant $\tau$ greater than $\max\{\tau'/C,\tau_0\}$. 

This concludes the proof of \eqref{eq:c3cecc}, and therefore the proof of Lemma~\ref{lem:avgdeg}.
\end{proof}

We are now ready to prove Lemma~\ref{lem:conrltwo}, which we restate here for convenience.
\begin{lemconrltwo}
\statelemconrltwo
\end{lemconrltwo}
\begin{proof}
Let $\epsilon,c>0$ be the constants in Lemma~\ref{lem:bigexpbigexp} and $M>0$ be the constant in Lemma~\ref{lem:avgdeg}.

By taking a union bound over the events in Lemmas~\ref{lem:numedges},~\ref{lem:bigexpbigexp} and~\ref{lem:avgdeg}, we obtain that there exists a constant integer $\ell$ such that the following hold for all sufficiently large $n$ and an arbitrary $\ell$-partition of the torus with probability $1-O(1/n^2)$ over the choice of $G=(V,E)\sim \Gc_{n,r}$:
\begin{enumerate}
\item \label{it:edges1} $2N\leq |E|\leq 3N$.
\item \label{it:edd33} every box-like set $S$, which is connected in $G$ and satisfies $100\ell^2\log n\leq |S|\leq 99N/100$, is  also $(\epsilon,c)$-expanding.
\item \label{it:avgdeg42} every set $S\subseteq V$ in $G$ satisfies $\sum_{v\in S}d_v\leq M|S|+|\partial S|$.
\end{enumerate}
Let $G$ be an arbitrary graph satisfying Items~\ref{it:edges1},~\ref{it:edd33} and~\ref{it:avgdeg42}. We only need to show that $G$ satisfies the conclusions of the lemma (for some appropriate constants $\rho,\tau,\chi>0$). Note, we may assume w.l.o.g. that $\epsilon,c\in (0,1/100)$, since the property that a set $S$ is $(\epsilon,c)$-expanding is maintained when decreasing the values of $\epsilon,c$. Similarly, we may assume w.l.o.g. that $M>100$ since Item~\ref{it:avgdeg42} is maintained when we increase the value of $M$.

Let $\eta\in (0,1)$ be a small constant so that $4\ell^2\eta\leq \epsilon/2$ and $\eta/(4\ell^2)\leq c$. By Lemma~\ref{lem:torusdecom2}, we have that  for every set $S\subseteq V$ it holds that 
\begin{equation}\label{eq:torusdecom2}
\big|\boxcore(S)\big|\geq (1-\eta)|S|, \mbox{ or else } |\partial^* S|\geq \eta|S|/(4\ell^2).
\end{equation}
We will show that
\begin{equation}\label{eq:finedex35}
|\partial S|/|S|\geq \eta/(4\ell^2) \mbox{ for any connected set $S$ satisfying $100\ell^2\log n\leq|S|\leq 98N/100$.}
\end{equation}
Indeed, let $S'=\boxcore(S)$. If $|S'|<(1-\eta)|S|$, we have by \eqref{eq:torusdecom2} that $|\partial S|\geq |\partial^* S|\geq \eta|S|/(4\ell^2)$, as needed. Hence, we may focus on the case that $|S'|\geq (1-\eta)|S|$. Let $S''$ be the smallest box-like set such that $S\subseteq S''$. Note that the set $S''$ is obtained from $S$ by just ``filling-up" those boxes $U$ in the $\ell$-partition which only partially intersect $S$ (i.e., $U\cap S\neq \emptyset$ and $U\cap S\neq S$). Therefore, the set $S''$ is a box-like set which is connected in $G$ and further satisfies
\[|S''|\leq |S'|+4\ell^2 |S\backslash S'|\leq |S|+4\ell^2\eta |S|\leq (1+\epsilon/2)|S|\leq 99N/100.\]
Therefore, by Item~\ref{it:edd33} above, we have that $S''$ is $(\epsilon,c)$-expanding. Since $S\subseteq S''$ and $|S|\geq (1-\epsilon)|S''|$, by the definition of an $(\epsilon,c)$-expanding set, we have that
\[|\partial S|\geq |\partial S\cap \partial S''|\geq c|S|\geq \eta|S|/(4\ell^2).\]
This finishes the proof of \eqref{eq:finedex35}.

For a connected set $S$ satisfying $|S|\leq 100\ell^2\log n$, by considering just the edges of the torus, we have by Theorem~\ref{thm:torusexpansion} that $|\partial S|\geq |\partial^* S|\geq 2\sqrt{|S|}$ and hence 
\begin{equation}\label{eq:finedex36}
|\partial S|/|S|\geq 2/\sqrt{|S|} \mbox{ for any connected set $S$ satisfying $|S|\leq 100\ell^2\log n$.}
\end{equation}

We are now ready to bound the conductance $\Phi(S)$ for a connected set $S$ with $\pi(S)\leq 1/2$. The assumption $\pi(S)\leq 1/2$ gives that $\sum_{v\in S}d_v\leq |E|\leq 3N$, where the last inequality holds from Item~\ref{it:edges1}. Since $d_v\geq 4$ for any $v\in V$, we have that $|S|\leq 3N/4$. Moreover, we have that
\begin{equation}\label{eq:phisrl2}
\Phi(S)=\frac{|\partial S|}{\frac{1}{2|E|}\big(\sum_{v\in S}d_v\big)\big(\sum_{v\notin S}d_v\big)}\geq \frac{|\partial S|}{\sum_{v\in S}d_v}, 
\end{equation}
where in the inequality we used the trivial bound $\sum_{v\notin S}d_v\leq 2|E|$. Further, for any set $S\subseteq V$ we have by Item~\ref{it:avgdeg42} that
\begin{equation}\label{eq:wq2}
\sum_{v\in S}d_v\leq M|S|+|\partial S|.
\end{equation}
It follows that 
\[\Phi(S)\geq \frac{|\partial S|}{\sum_{v\in S}d_v}\geq \frac{\frac{|\partial S|}{|S|}}{M+\frac{|\partial S|}{|S|}}.\]
Combining this bound with \eqref{eq:finedex35} and \eqref{eq:finedex36}, we obtain the bounds in the lemma with $\rho=100\ell^2$ and $\tau=\min\{\frac{\eta}{4\ell^2}/(M+\frac{\eta}{4\ell^2}),1/M\}$. This finishes the proof of the first part.

For the second part, consider a connected set $S$ such that $\Phi(S)<\tau$ and $\pi(S)\leq 1/2$. Our goal is to show that there exists a constant $\chi>0$ such that 
\begin{equation}\label{eq:boundonpi}
\frac{\sum_{v\in S}d_v}{2|E|}\leq \chi |S|/N.
\end{equation}
By the first part of the lemma, we know that $|S|\leq \rho \log n$. Combining the bound in \eqref{eq:phisrl2} with the inequality $\Phi(S)<\tau$, we also obtain that $|\partial S|\leq \tau\sum_{v\in S}d_v$. Then, using \eqref{eq:wq2}, we obtain that  
\[\sum_{v\in S}d_v\leq M|S|+\tau \sum_{v\in S}d_v, \]
so, using  the fact that $|E|\geq 2N$, we get that
\[\frac{\sum_{v\in S}d_v}{2|E|}\leq \frac{M}{1-\tau} |S|/(4N),\]
which proves \eqref{eq:boundonpi} with $\chi=M/(4(1-\tau))$. 

This concludes the proof of Lemma~\ref{lem:conrltwo}.
\end{proof}

\section{Upper bound on the mixing time for $r=2$}\label{sec:req2}
In this section we establish the upper bound claimed in Theorem~\ref{thm:main} for the small-world network for $r=2$.
In particular, we will show that with probability $1-O(1/n^2)$ over the choice of
the graph $G\sim\Gc_{n,r}$, the lazy random walk on the small-world network model has
mixing time $\Tmix= O((\log{n})^4)$.

\subsection{Proof Outline}
Let $G=(V,E)\sim\Gc_{n,r}$. For a subset $S\subseteq V$, we will use $\partialv S$ to denote the set of vertices in $V\backslash S$ which are connected by an edge to a vertex in $S$. Note that $|\partial S|\geq |\partialv S|$, so a lower bound on $|\partialv S|$ also implies a lower bound on $|\partial S|$ (recall, $\partial S$ is the set of edges with exactly one endpoint in $S$). 
The following lemma will be used to obtain a lower bound on $|\partialv S|$ for sets $S\subseteq V$. Throughout this section, we will use $\alpha$ to denote $|S|/N$; although $\alpha$ is a function of $S$, we suppress this from the notation because the set $S$ will be clear from context.
\newcommand{\statelemarbexpreqtwo}{Let $r=2$. Then, there exists a constant $c>0$ such that for all sufficiently large integers $n$, the following holds.

For all sets $S\subseteq V$ with $|S|=\alpha  N$ and $\alpha\in(0,1]$, it holds that 
\[\mbox{$\Pr_{\Gc_{n,r}}$}\bigg(|\partialv S|\leq c|S|\frac{\log (1/\alpha)}{\log n}\bigg)\leq \exp\Big(-c |S|\frac{\log (1/\alpha)}{\log n}\Big).\]
}
\begin{lemma}\label{lem:arbexpreqtwo}
\statelemarbexpreqtwo
\end{lemma}
Lemma \ref{lem:arbexpreqtwo} suggests that larger sets have worse vertex expansion for $r=2$: for $|S|=N^{1-\Omega(1)}$, the vertex expansion 
 is linear in  $|S|$, while for $|S|=\Omega(N)$ the vertex expansion  is roughly $|S|/\log n$. Using Lemma~\ref{lem:arbexpreqtwo}, we  show the following in Section~\ref{sec:bigexpansiontwo}.

\newcommand{\statelembigexpansiontwo}{Let $r=2$. There exist constants $c, \ell_0>0$ such that for all sufficiently large integers $n$ and $\ell=\left\lceil \ell_0(\log n)^{1/2}\right\rceil$, the following holds for any $\ell$-partition of the torus $T=(B_n,E_n)$. 

With probability $1-O(1/n^2)$ over the choice of $G=(V,E)\sim \Gc_{n,r}$, every box-like set $S$ with $|S|=\alpha N$ and $\alpha\in(0,99/100]$ satisfies $|\partialv S|\geq c|S|\frac{\log(1/\alpha)}{\log n}$.
}
\begin{lemma}\label{lem:bigexpansiontwo}
\statelembigexpansiontwo
\end{lemma}

The following lemma is the analogue of Lemma~\ref{lem:avgdeg} ($r<2$), and allows us to control the probability mass of a set $S$ in the stationary distribution of the random walk. The difference for $r=2$ is that the expected number of edges in the interior of a set $S$ is significantly higher than in the case $r<2$. Nevertheless, we can recover the following if we restrict our attention to connected sets $S$. The proof is analogous to  the proof of \cite[Lemma 8]{AL} and requires some care to handle the dependency of the average degree of $S$ with the event that $S$ is connected.
\newcommand{\statelemavgdegtwo}{Let $r=2$. There exists a constant $M>0$ such that the following holds for all sufficiently large $n$ with probability $1-O(1/n^2)$ over the choice of the graph $G=(V,E)\sim\Gc_{n,r}$. 

For every set $S\subseteq V$ which is connected in $G$, it holds that $\sum_{v\in S}d_v\leq M\max\{|S|,\log n\}$.
}
\begin{lemma}\label{lem:avgdegtwo}
\statelemavgdegtwo
\end{lemma}
The proof of Lemma~\ref{lem:avgdegtwo} is given in Section~\ref{sec:avgdegtwo}. Combining Lemmas~\ref{lem:bigexpansiontwo} and~\ref{lem:avgdegtwo}, we obtain the following bounds on the conductance of the random walk.

\newcommand{\statelemconreqtwo}{Let $r=2$. There exists a constant $\tau>0$ such that the following holds for all sufficiently large integers  $n$ with probability $1-O(1/n^2)$ over the choice of the graph $G\sim\Gc_{n,r}$. 

For every connected set $S$ in $G$ with $|S|=\alpha N$ and $\pi(S)\leq 1/2$, the conductance $\Phi(S)$ satisfies 
\[\Phi(S)\geq \tau \frac{\log (1/\alpha)}{(\log n)^2}.\]
}
\begin{lemma}\label{lem:conreqtwo}
\statelemconreqtwo
\end{lemma}
The proof of Lemma~\ref{lem:conreqtwo} is given in Section~\ref{sec:conreqtwo}.
\begin{corollary}\label{cor:req2}
Let $r=2$ and $n$ be a sufficiently large integer. With probability $1-O(1/n^2)$ over the choice of
the graph $G\sim\Gc_{n,r}$, the lazy random walk on $G$
 satisfies $\Tmix=O((\log n)^4)$.
\end{corollary}
\begin{proof}
We proceed as in the proof of Corollary~\ref{cor:rl2}. Let $G=(V,E)\sim\Gc_{n,r}$. By a union bound, we have that with probability $1-O(1/n^2)$ the graph $G$ satisfies Lemmas~\ref{lem:numedges} and~\ref{lem:conreqtwo}. For all such graphs $G$, we will show that  $\Tmix=O((\log n)^4)$.

Recall from Theorem~\ref{thm:FR} that there exists an absolute constant $C>0$ such that 
\begin{equation*}\tag{\ref{eq:phij12}}
\Tmix\leq C \sum^{\left\lceil \log_2(\pi_{\mathrm{min}}^{-1})\right\rceil}_{j=1}\bigg(\frac{1}{\widetilde{\Phi}(2^{-j})}\bigg)^2.
\end{equation*}
where $\widetilde{\Phi}(p):=\min\{\Phi(S) \mid S\mbox{ is connected},\, p/2\leq \pi(S)\leq p\}$, and $\widetilde{\Phi}(p)=1$ if the minimization is over an empty set.  By Lemma~\ref{lem:numedges}, we have that $|E|\leq 3N$ and hence $\pi_{\mathrm{min}}\geq 4/(2|E|)\geq 2/(3N)$, so that $\log_2(1/\pi_{\mathrm{min}})=O(\log n)$.

We will show that
\begin{equation}\label{eq:4rv4v3v3b45y}
\mbox{for all $j=1,\hdots,\left\lceil \log_2(1/\pi_{\mathrm{min}})\right\rceil$, it holds that } \widetilde{\Phi}(2^{-j})\geq (\tau/10)\frac{j}{(\log n)^2},
\end{equation}
where $\tau$ is the constant in Lemma~\ref{lem:conreqtwo}. Combining \eqref{eq:4rv4v3v3b45y} with \eqref{eq:phij12} and the fact that  $\sum_{j\geq 1}\frac{1}{j^2}=O(1)$ yields that $\Tmix=O((\log n)^4)$.

To prove \eqref{eq:4rv4v3v3b45y}, consider an arbitrary index $j$ with $1\leq j\leq \left\lceil \log_2(1/\pi_{\mathrm{min}})\right\rceil$. If $\widetilde{\Phi}(2^{-j})=1$, the inequality in \eqref{eq:4rv4v3v3b45y} is trivially true (for all sufficiently large integers $n$), so we may assume that $\widetilde{\Phi}(2^{-j})<1$. In particular, there exists a connected set $S$ such that $2^{-j-1}\leq \pi(S)\leq 2^{-j}$. Consider an arbitrary such connected set $S$ and let $\alpha=|S|/N$. Since $\pi(S)\leq 2^{-j}$, we have that 
\[\mbox{$\sum_{v\in S}d_v\leq 2^{-j+1}|E|$ and therefore $|S|\leq 3\cdot 2^{-j-1}N$ (using that $d_v\geq 4$ and $|E|\leq 3N$).}\]
It follows that  $\alpha\leq 3\cdot 2^{-j-1}$ and hence $\log(1/\alpha)\geq (j+1)\log 2-\log 3\geq j/10$. By Lemma~\ref{lem:conreqtwo}, we therefore have that $\Phi(S)\geq(\tau/10)\frac{j}{(\log n)^2}$. Since $S$ was an arbitrary connected set with $2^{-j-1}\leq \pi(S)\leq 2^{-j}$, we obtain that $\widetilde{\Phi}(2^{-j})\geq (\tau/10)\frac{j}{(\log n)^2}$, as needed.

This concludes the proof of \eqref{eq:4rv4v3v3b45y} and therefore the proof of Corollary~\ref{cor:req2}.
\end{proof}

\subsection{Lower bounding the expansion -- Proof of Lemma~\ref{lem:arbexpreqtwo}}\label{sec:arbexpreqtwo}
In this section, we prove Lemma~\ref{lem:arbexpreqtwo}, which we restate here for convenience.
\begin{lemarbexpreqtwo}
\statelemarbexpreqtwo
\end{lemarbexpreqtwo}
\begin{proof}[Proof]
Consider an arbitrary set $S\subseteq V$ with $|S|=\alpha N$. 

For $w\in V$, let $X_w$ be the indicator r.v. that there is an edge from $w$ to the set $S$. Note that 
\begin{equation}\label{lem:partXXtwo}
|\partialv S|\geq X_S, \mbox{ where } X_S:=\sum_{w\in V\backslash S}X_w.
\end{equation}
For $u,w\in V$, let $Y_{u,w}$ be the indicator r.v. that there is an edge from $u$ to $w$ in $G$.  Note that the random variables $\{X_w\}_{w\in V\backslash S}$ are independent since each $X_w$ is determined by the random variables $\{Y_{u,w}\}_{u\in S}$, where $Y_{u,w}$ is the indicator r.v. that there is an edge between vertices  $u$ and $w$ in $G$. Thus, for $\mu_S:=\E_{\Gc_{n,r}}[X_S]$, we obtain by Lemma~\ref{lem:chernoff} that
\begin{equation}\label{eq:4v444c2}
\mbox{$\Pr_{\Gc_{n,r}}$}\Big(|\partialv S|\leq \tfrac{1}{2}\mu_S\Big)\leq \exp\big(-\tfrac{1}{10}\mu_S\big).
\end{equation}
Our goal will be therefore to give a lower bound on $\mu_S$.

We start by showing the bound
\begin{equation}\label{eq:f4323f}
\mu_S\geq \frac{1}{6}m_S, \mbox{ where } m_S:=\frac{1}{Z}\sum_{u\in S}\sum_{w\in V\backslash S}\dist{u}{w}^{-2},
\end{equation}
where $Z$ is the normalising factor given in \eqref{eq:defZ} and $\dist{u}{w}$ denotes the distance between the vertices $u,w$ in the torus. To see \eqref{eq:f4323f}, note first that for all $w\in V\backslash S$, we have that 
\[\mbox{$\sum_{u\in S} \E_{\Gc_{n,r}}[Y_{u,w}]\leq \sum_{u\in V} \E_{\Gc_{n,r}}[Y_{u,w}]=5$},\]  
since $w$ has four neighbours from the torus and one long-range neighbour in expectation. Using that $1-x\leq \e^{-x}\leq 1-x/6$ for $x\in [0,5]$, we therefore obtain that
\begin{align*}
\E_{\Gc_{n,r}}[X_w]&=1-\prod_{u\in S}\mbox{$\Pr_{\Gc_{n,r}}$}(Y_{u,w}=0)=1-\prod_{u\in S}\Big(1-\E_{\Gc_{n,r}}[Y_{u,w}]\Big)\\
&\geq 1-\exp\bigg(-\sum_{u\in S} \E_{\Gc_{n,r}}[Y_{u,w}]\bigg)\geq \frac{1}{6}\sum_{u\in S} \E_{\Gc_{n,r}}[Y_{u,w}].
\end{align*}
For all sufficiently large $n$, we claim that $\E_{\Gc_{n,r}}[Y_{u,w}]\geq \dist{u}{w}^{-r}/Z$: for non-adjacent vertices $u,w$ in the torus, the inequality holds at equality by the definition of the model $\Gc_{n,r}$; for vertices $u,w$ which are adjacent  in the torus, we have $\E_{\Gc_{n,r}}[Y_{u,w}]=1\geq 1/Z$, where the last inequality holds for all sufficiently large $n$  by Lemma~\ref{lem:factor}. Combining this with the equality $\mu_S=\sum_{w\in V\backslash S} \E_{\Gc_{n,r}}[X_w]$, we obtain \eqref{eq:f4323f}.

Recall by Lemma~\ref{lem:factor} that, for $r=2$, the normalising factor $Z$ is bounded by $Z\leq C\log n$ for some absolute constant $C>0$. We will show that for all sufficiently large $n$ it holds that 
\begin{equation}\label{eq:everf2ef334}
m_S\geq \frac{\log(1/\alpha)}{\log n}|S|/(50 C) \mbox{ for $\alpha\in[10^{-4},1]$}, \quad  m_S\geq \frac{\log(1/\alpha)}{\log n}|S|/(2 C)\mbox{ for $\alpha\in(0,10^{-4})$}.
\end{equation}
In light of \eqref{eq:4v444c2}, \eqref{eq:f4323f} and \eqref{eq:everf2ef334}, we therefore obtain the statement of the lemma (with $c=1/(3000C)$).  We thus focus on proving \eqref{eq:everf2ef334}. In the following argument, we use in our calculations that the value of $r$ is equal to 2.

For $\alpha\in[10^{-4},1]$, we use the trivial bounds $\dist{u}{w}\leq 2n$ and $N\geq n^2$ and the fact that $|V\backslash S|= (1-\alpha)N$ to obtain that 
\begin{align*}
m_S&=\frac{1}{Z}\sum_{u\in S}\sum_{w\in V\backslash S}\dist{u}{w}^{-2}\geq \sum_{u\in S}\frac{(1-\alpha) N}{Z}(2n)^{-2}\\
&\geq \frac{1-\alpha}{\log n}|S|/(4C)\geq \frac{\log(1/\alpha)}{\log n}|S|/(50 C),
\end{align*}
where in the last inequality we used that $1-\alpha\geq \frac{1}{10}\log(1/\alpha)$ for all $\alpha\in[10^{-4},1]$. This proves \eqref{eq:everf2ef334} in the case $\alpha\in[10^{-4},1]$.

For $\alpha\in(0,10^{-4})$, observe that for every $u\in S$ it holds that 
$Z\leq \sum_{w\in V; w\neq u}\dist{u}{w}^{-2}$, so 
\begin{equation}\label{eq:mupS3r3}
m_S=\frac{1}{Z}\sum_{u\in S}\sum_{w\in V\backslash S}\dist{u}{w}^{-r}\geq |S|-\frac{1}{Z}\sum_{u\in S}\sum_{w\in S; w\neq u}\dist{u}{w}^{-2}.
\end{equation}
Fix an arbitrary vertex $u\in S$, and let
\begin{equation*}
Q_u:=\sum_{w\in S; w\neq u}\dist{u}{w}^{-2}.
\end{equation*}
To upper bound the quantity $Q_u$, note that we can only increase its value by moving all vertices in $S\backslash\{u\}$ as close as possible to $u$. In particular, let $\ell_0:=\left\lceil 3\sqrt{\alpha} n\right\rceil< n$ and consider the set of vertices 
\[U=\{w\in V\mid  \dist{u}{w}\leq \ell_0\}.\]
Note that for $\ell\leq \ell_0\leq  n$, the  number of vertices at distance $\ell$ from $v$ is $4\ell$, and hence
\[|U|=1+4\sum^{\ell_0}_{i=1}\ell=1+2\ell_0(\ell_0+1)\geq  1+18\alpha n^2\geq 1+\alpha N.\] 
It follows that $|U\backslash \{u\}|\geq |S|$ and therefore
\begin{equation}\label{eq:QvUwe}
Q_u\leq \sum_{w\in U; w\neq u}\dist{u}{w}^{-2}\leq Z+4-\sum_{w\notin U}\dist{u}{w}^{-2}.
\end{equation}
Thus, it remains to lower bound $\sum_{w\notin U}\dist{u}{w}^{-2}$. For all distances $\ell$ satisfying $\ell_0<\ell\leq  n$, there are $4\ell$ vertices at distance $\ell$ from $v$. It follows that (cf. \eqref{eq:Hnboundrew})
\begin{equation}\label{eq:wUvwtau}
\sum_{w\notin U}\dist{u}{w}^{-2}\geq 4\sum^{n}_{\ell =\ell_0+1}\frac{1}{\ell}\geq 4\int^{n}_{\ell_0+1} \frac{1}{x} dx\geq 4\log (n/\ell_0)\geq \frac{1}{2}\log(1/\alpha)+4,
\end{equation}
where the last  inequality follows from the assumption $\alpha\in(0,10^{-4})$ and the bound $\ell_0\leq 6\sqrt{\alpha} n$ (note that $\alpha N=|S|\geq 1$ and hence $3\sqrt{\alpha} n\geq 1$). Combining \eqref{eq:QvUwe} and \eqref{eq:wUvwtau}, and using that $Z\leq C\log n$, we thus obtain that for all $u\in S$, it holds that
\[Q_u\leq Z -\frac{Z}{2C}\frac{\log(1/\alpha)}{\log n}\]
Plugging this into \eqref{eq:mupS3r3} yields that $m_S\geq \frac{\log(1/\alpha)}{\log n}|S|/(2 C)$, thus completing the proof of \eqref{eq:everf2ef334} in the case $\alpha\in(0,10^{-4})$ as well. This concludes the proof of Lemma~\ref{lem:arbexpreqtwo}.
\end{proof}

\subsection{The vertex expansion of box-like sets -- Proof of Lemma~\ref{lem:bigexpansiontwo}}\label{sec:bigexpansiontwo}
In this section, we prove Lemma~\ref{lem:bigexpansiontwo}, which we restate here for convenience.
\begin{lembigexpansiontwo}
\statelembigexpansiontwo
\end{lembigexpansiontwo}
\begin{proof}
Let $G=(V,E)\sim \Gc_{n,r}$. By Lemma~\ref{lem:arbexpreqtwo}, there exists a constant $c'>0$ such that all sets $S\subseteq V$ with $|S|=\alpha N$ and $\alpha\in(0,1]$ satisfy
\begin{equation}\label{eq:bigexpansion2}
\mbox{$\Pr_{\Gc_{n,r}}$}\bigg(|\partialv S|\leq c' |S|\frac{\log (1/\alpha)}{\log n}\bigg)\leq \exp\Big(-c' |S|\frac{\log (1/\alpha)}{\log n}\Big).
\end{equation}
Decreasing the value of $c'$ does not affect the validity of \eqref{eq:bigexpansion2}, so we may assume w.l.o.g. that $c'$ is a constant in the interval $(0,1)$.

Let $\ell_0=100/c'$ and $\ell=\left\lceil \ell_0(\log n)^{1/2}\right\rceil$. Consider an arbitrary $\ell$-partition $\mathcal{U}=\{U_1,\hdots, U_Q\}$ of the torus. Note that  $N/(4\ell^2)\leq Q\leq N/\ell^2$ (since every box  $U\in\mathcal{U}$ contains at least $\ell^2$ and at most $4\ell^2$ vertices). For a box-like set $S$, denote by $q_S$ the number of boxes $U\in\mathcal{U}$ such that $S\cap U\neq \emptyset$.

Consider an arbitrary box-like set $S$ such $|S|\leq 99N/100$. Let $\alpha:=|S|/ N$. Since $N\geq Q\ell^2\geq (10^4/c')Q\log n$ and $\log(100/99)\geq 1/100$, \eqref{eq:bigexpansion2} gives for $\alpha\in[1/100,99/100]$ that
\begin{equation*}
\mbox{$\Pr_{\Gc_{n,r}}$}\bigg(|\partialv S|\leq  \frac{c'|S|}{100\log n}\bigg)\leq \exp\Big(-\frac{c'\alpha N}{100\log n}\Big)\leq \exp\big(-100\alpha Q\big)\leq \exp(-Q).
\end{equation*}
Since the total number of box-like sets is at most $2^Q$, by a union bound we obtain that with probability $1-\e^{-\Omega(n)}$, every box-like set with $|S|=\alpha N$ and $\alpha\in [1/100,99/100]$ satisfies 
\begin{equation}\label{eq:bigexpansionboz2}
|\partialv S|\geq  \frac{c'|S|}{100\log n}\geq \frac{c'|S|}{1000} \frac{\log(1/\alpha)}{\log n},
\end{equation}
where in the last inequality we used that for $\alpha\in[1/100,99/100]$ it holds that $10\geq \log(1/\alpha)$.

Consider now an arbitrary box-like set $S$ such that $|S|\leq N/100$ and let $\alpha:=|S|/N$, so that $\alpha\in (0,1/100]$. Since $S$ is box-like, for every box $U$ such that $|S\cap U|\neq 0$,  we have that $|S\cap U|=|U|\geq \ell^2$ and $|S\cap U|=|U|\leq 4\ell ^2$. Thus, with $q:=q_S$, we obtain that
\begin{equation}\label{eq:4v445tb5f3}
4(q/Q)\geq \alpha\geq (q/Q)/4, \mbox{ which also yields that } q/Q\leq 1/10.
\end{equation}
Using once again that $N\geq Q\ell^2\geq (10^4/c')Q\log n$, \eqref{eq:bigexpansion2} gives  that
\begin{equation}\label{eq:bigexpansionboz3}
\mbox{$\Pr_{\Gc_{n,r}}$}\bigg(|\partialv S|\leq c' |S|\frac{\log (1/\alpha)}{\log n}\bigg)\leq \exp\Big(-10^4\alpha  Q\log (1/\alpha)\Big)\leq \exp\Big(-10q\log (Q/q)\Big).
\end{equation}

For $1\leq q\leq Q/10$, let $\mathcal{E}_q$ be the event that there exists a box-like set $S$ with $q_S=q$ such that $|\partialv S|\leq c' |S|\frac{\log (1/\alpha)}{\log n}$, where $\alpha=|S|/N$. The number $W_q$ of box-like sets $S$ with $q_S=q$ is bounded by
\begin{equation}\label{eq:wqbound2}
W_q=\binom{Q}{q}\leq \Big(\frac{\e Q}{q}\Big)^{q} \leq \exp\Big(2 q \log(Q/q)\Big).
\end{equation}
Combining \eqref{eq:bigexpansionboz3} and \eqref{eq:wqbound2}, we obtain that 
\begin{equation}\label{eq:5b56hb6b4}
\mbox{$\Pr_{\Gc_{n,r}}$}(\mathcal{E}_q)\leq \e^{-8q \log (Q/q)}\leq \e^{-8\log Q},
\end{equation}
where in the last inequality we used that the function $x\log(Q/x)$ is concave for $1\leq x\leq Q$ and therefore $\min_{q\in [1,Q/10]}q \log (Q/q)=\min\big\{\log Q,(Q\log 10)/10\big\}=\log Q$.

Let $\mathcal{E}$ be the event that there exists a box-like set with $|S|=\alpha N$ and $\alpha\in (0,1/100]$ such that $|\partialv S|\geq  c'|S|\frac{\log(1/\alpha)}{\log n}$. Then, by \eqref{eq:4v445tb5f3}, \eqref{eq:5b56hb6b4} and a union bound, we have that 
\begin{equation}\label{eq:bigexpansionboz5}
\mbox{$\Pr_{\Gc_{n,r}}$}(\mathcal{E})\leq \mbox{$\Pr_{\Gc_{n,r}}$}\bigg(\bigcup_{q\in[1,Q/10]} \mathcal{E}_q\bigg)\leq Q\e^{-8 \log Q}\leq 1/n^2.
\end{equation}

A union bound over the events in \eqref{eq:bigexpansionboz2} and \eqref{eq:bigexpansionboz5} yields the statement of the lemma with the constant $c=c'/1000>0$.
\end{proof}

\subsection{Upper bounding the average degree -- Proof of Lemma~\ref{lem:avgdegtwo}}\label{sec:avgdegtwo}
In this section, we prove Lemma~\ref{lem:avgdegtwo}. Our proof follows closely the proof of \cite[Lemma 8]{AL} and we give it for the sake of self-containedness. 

The proof uses Harris' inequality for product measures, which is a special case of the FKG inequality. Let $E=\{e_1,\hdots,e_n\}$ be  a finite set and consider the probability space induced by choosing a random subset of $E$ by including each element $e_i$ with probability $p_i$ ($0\leq p_i\leq 1$). An event $\mathcal{F}$ is called increasing if, for all $A,B\subseteq E$ with $A\subseteq B$,  $A\in \mathcal{F}$ implies that $B\in \mathcal{F}$. Similarly, 
an event $\mathcal{F}$ is called decreasing if, for all $A,B\subseteq E$ with $A\subseteq B$,  $B\in \mathcal{F}$  implies that $A\in \mathcal{F}$.
\begin{lemma}[{\cite{Harris,FKG}}]\label{lem:harris}
If $\mathcal{F}$ is an increasing event and $\mathcal{E}$ is a decreasing event, then $\Pr(\mathcal{F}\mid \mathcal{E})\leq \Pr(\mathcal{F})$.
\end{lemma}

We use Lemma~\ref{lem:harris} to show the following for the small-world graph  $G=(V,E)\sim \Gc_{n,r}$.
\begin{lemma}\label{lem:c1c2fwq}
Let $r=2$ and $n$ be an integer. Then, for all sets $S\subseteq V$, it holds that 
\[\mbox{$\Pr_{\Gc_{n,r}}$}\Big(\mbox{$\sum_{v\in S}$}\,d_v\geq 150\max\{|S|,\log n\}\,\Big|\, \mbox{$S$ is connected}\Big)\leq \exp\big(-30 \max\{|S|,\log n\}\big).\]
\end{lemma}
\begin{proof}
Let $G=(V,E)\sim \Gc_{n,r}$ and consider an arbitrary set $S\subseteq V$. 

For vertices $u,w\in V$, let $Y_{u,w}$ be the indicator r.v. that there is a long-range edge between $u,w$ in $G$. Let $\binom{S}{2}:=\big\{\{u,w\}\mid u,w\in S, u\neq w\big\}$  denote the set of all unordered pairs of vertices in $S$ and define
\[Y_{in}:=\sum_{\{u,w\}\in \binom{S}{2}} Y_{u,w}, \quad Y_{out}:=\sum_{u\in S, w\in V\backslash S}Y_{u,w}.\] 
Observe that $\sum_{v\in S}d_v\leq 4|S|+2Y_{in}+Y_{out}$, so to prove the lemma we only need to show that 
\begin{align}
\mbox{$\Pr_{\Gc_{n,r}}$}\big(Y_{out}\geq 40U\,\big|\, \mbox{$S$ is connected}\big)&\leq \exp\big(-40 U\big),\label{eq:frr3f3}\\
\mbox{$\Pr_{\Gc_{n,r}}$}\big(Y_{in}\geq 40U\,\big|\, \mbox{$S$ is connected}\big)&\leq \exp\big(-40 U\big),\label{eq:rtet}
\end{align}
where $U:=\max\{|S|,\log n\}$. Note that the event that $S$ is connected is determined by the random variables $\{Y_{u,w}\}_{\{u,w\}\in \binom{S}{2}}$ and therefore it is independent of the event $Y_{out}\geq 40U$. Since every vertex adds in expectation one long-range neighbour, we have the bound $\E_{\Gc_{n,r}}[Y_{out}]\leq |S|\leq U$ and therefore Lemma~\ref{lem:chernoff} yields \eqref{eq:frr3f3}.

The proof of \eqref{eq:rtet} requires more work. Let $\mathcal{E}$ be the event that $S$ is connected in $G$. Following \cite{AL}, we decompose $\mathcal{E}$ as follows. Let $t_1,\hdots, t_h$ be an enumeration of all labelled (spanning) trees on the vertex set $S$. For $i\in [h]$, let $E_i$ be the edge set of the tree $t_i$ and $\mathcal{E}_i$ be the event that $E_i\subseteq E$. Further, let $\mathcal{E}_i'$ be the event $\mathcal{E}_i\backslash \bigcup_{j<i} \mathcal{E}_j$ and note that  $\{\mathcal{E}_i'\}_{i\in [h]}$ is a partition of the event $\mathcal{E}$. Observe also that for disjoint events $A,B$ and an event $C$, it holds that $\Pr(C\mid A\cup B)\leq \max\{\Pr(C\mid A), \Pr(C\mid B)\}$, so we obtain that 
\begin{equation}\label{eq:v3rv3g33d3f}
\mbox{$\Pr_{\Gc_{n,r}}$}\big(Y_{in}\geq 40U\,\big|\, \mathcal{E})\leq \max_{i\in [h]} \mbox{$\Pr_{\Gc_{n,r}}$}\big(Y_{in}\geq 40U\,\big|\, \mathcal{E}_i'\big).
\end{equation}
Now, using Harris' inequality (Lemma~\ref{lem:harris}) for the distribution $\mbox{$\Pr_{\Gc_{n,r}}$}(\cdot \mid \mathcal{E}_i)$ (note that this is also a product distribution since it only conditions the edges in $E_i$ to be in $E$), we  obtain that for all $i\in [h]$, it holds that 
\begin{equation}\label{eq:vrv34gb354htb3}
\begin{aligned}
\mbox{$\Pr_{\Gc_{n,r}}$}\big(Y_{in}\geq 40U\,\big|\, \mathcal{E}_i'\big)&=\mbox{$\Pr_{\Gc_{n,r}}$}\Big(Y_{in}\geq 40U\,\Big|\, \big(\mbox{$\bigcup$}_{j<i}\, \mathcal{E}_j\big)^c \cap \mathcal{E}_i\Big)\\
&\leq \mbox{$\Pr_{\Gc_{n,r}}$}\big(Y_{in}\geq 40U\,\big|\,\mathcal{E}_i\big).
\end{aligned}
\end{equation}
Consider an arbitrary $i\in [h]$. Conditioned on $\mathcal{E}_i$, $Y_{in}$ is distributed as 
\[X:=|S|-1+\sum_{\{u,w\}\in \binom{S}{2}\backslash E_i} X_{u,w},\] 
where $\{X_{u,w}\}_{\{u,w\}\in \binom{S}{2}\backslash E_i}$ are independent $\{0,1\}$ random variables with $\Pr(X_{u,w}=1)=\frac{1}{Z}\dist{u}{w}^{-2}$. It follows that $\E[X]\leq 2|S|\leq 2U$ and therefore  by Lemma~\ref{lem:chernoff} we obtain that   
\[\mbox{$\Pr_{\Gc_{n,r}}$}\big(Y_{in}\geq 40U\,\big|\,\mathcal{E}_i\big)\leq \exp\big(-40 U\big).\]
Combining this with \eqref{eq:v3rv3g33d3f} and \eqref{eq:vrv34gb354htb3}, we obtain \eqref{eq:rtet}, concluding the proof of Lemma~\ref{lem:c1c2fwq}.
\end{proof}

Using Lemma~\ref{lem:c1c2fwq}, we can now prove Lemma~\ref{lem:avgdegtwo} by a union bound over all connected sets.
\begin{lemavgdegtwo}
\statelemavgdegtwo
\end{lemavgdegtwo}
\begin{proof}
We will show the lemma with $M=150$. Let $G=(V,E)\sim \Gc_{n,r}$. By Lemma~\ref{lem:c1c2fwq}, all sets $S\subseteq V$ satisfy
\begin{equation}\label{eq:34g3423rf}
\mbox{$\Pr_{\Gc_{n,r}}$}\Big(\mbox{$\sum_{v\in S}$}\,d_v\geq 150\max\{|S|,\log n\}\,\Big|\, \mbox{$S$ is connected}\Big)\leq \e^{-30 \max\{|S|,\log n\}}.
\end{equation}

For an integer $q\geq 1$, let $\mathcal{E}_q$ be the event that there exists a set $S$ with $|S|=q$ which is connected in $G$ but $\sum_{v\in S}d_v\geq 150\max\{|S|,\log n\}$. To prove the lemma, it  suffices to show that 
\begin{equation}\label{eq:ggreg45g}
\mbox{$\Pr_{\Gc_{n,r}}$}\bigg(\bigcup_{1\leq q\leq N} \mathcal{E}_q\bigg)\leq 1/n^2
\end{equation}
By Lemma~\ref{lem:numcon} with $\ell=1$ (and the trivial partition where every vertex of the torus is a ``box"), we have that the number $W_q$ of sets $S$ with $|S|=q$ that are connected in $G$ satisfies the bound
\begin{equation}\label{eq:wqell42342}
\E_{\Gc_{n,r}}[W_q]\leq n^2 40^{q}.
\end{equation}
Let $\mathcal{F}_q$ denote the set of subsets $S\subseteq V$ with $|S|=q$. Then, using \eqref{eq:34g3423rf} and~\eqref{eq:wqell42342}, we have
\begin{align*}
\mbox{$\Pr_{\Gc_{n,r}}$}(\mathcal{E}_q)&\leq \sum_{S\in \mathcal{F}_q} \mbox{$\Pr_{\Gc_{n,r}}$}(\mbox{$S$ is connected})\, \mbox{$\Pr_{\Gc_{n,r}}$}\Big(\mbox{$\sum_{v\in S}$}\,d_v\geq 150\max\{|S|,\log n\}\,\Big|\, \mbox{$S$ is connected}\Big)\\
&\leq \e^{-30 \max\{|S|,\log n\}}\sum_{S\in \mathcal{F}_q} \mbox{$\Pr_{\Gc_{n,r}}$}(\mbox{$S$ is connected})=\e^{-30 \max\{|S|,\log n\}}\E_{\Gc_{n,r}}[W_q]\\
&\leq  n^2 \e^{-30 \max\{|S|,\log n\}}40^{q}\leq n^{-10}.
\end{align*}
By a union bound over the possible values of $q$, we therefore obtain \eqref{eq:ggreg45g}, as wanted.
\end{proof}

\subsection{Conductance bounds -- Proof of Lemma~\ref{lem:conreqtwo}}\label{sec:conreqtwo}
In this section, we prove Lemma~\ref{lem:conreqtwo}, which we restate here for convenience.
\begin{lemconreqtwo}
\statelemconreqtwo
\end{lemconreqtwo}
\begin{proof}
Let $c,\ell_0>0$ be the constants in Lemma~\ref{lem:bigexpansiontwo} and $M>0$ be the constant in Lemma~\ref{lem:avgdegtwo}. Let $\ell=\left\lceil \ell_0(\log n)^{1/2}\right\rceil$ and note that for all sufficiently large $n$ we have $\ell\leq 2\ell_0(\log n)^{1/2}$.

By taking a union bound over the events in Lemmas~\ref{lem:numedges},~\ref{lem:bigexpansiontwo} and~\ref{lem:avgdegtwo}, we obtain that the following hold for all sufficiently large $n$ and an arbitrary $\ell$-partition of the torus with probability $1-O(1/n^2)$ over the choice of $G=(V,E)\sim \Gc_{n,r}$:
\begin{enumerate}
\item \label{it:edg24} $2N\leq |E|\leq 3N$.
\item \label{it:eeff} every box-like set $S$ with $|S|=\alpha N$ and $\alpha\in(0,99/100]$ satisfies $|\partialv S|\geq c|S|\frac{\log(1/\alpha)}{\log n}$.
\item \label{it:avgd435} every set $S\subseteq V$ which is connected in $G$ satisfies $\sum_{v\in S}d_v\leq M\max\{|S|,\log n\}$.
\end{enumerate}
Let $G$ be an arbitrary graph satisfying Items~\ref{it:edg24},~\ref{it:eeff} and~\ref{it:avgd435}. We only need to show that $G$ satisfies the conclusions of the lemma (for some appropriate constant $\tau>0$). Note, we may assume w.l.o.g. that $c\in (0,1/100)$ and $M>100$, since Items~\ref{it:eeff} and~\ref{it:avgd435} continue to hold when we decrease the value of $c$ and increase the value of $M$.

By Lemma~\ref{lem:torusdecom2}, we have that  for every set $S\subseteq V$ and any $\eta\in (0,1)$ it holds that 
\begin{equation}\label{eq:torusdecom4t3}
\big|\boxcore(S)\big|\geq (1-\eta)|S|, \mbox{ or else } |\partial^* S|\geq \eta|S|/(4\ell^2).
\end{equation}
We will first show that for the constant $\xi=c/(10^{4}\ell_0^2)>0$, it holds that
\begin{equation}\label{eq:finedexge44g4}
|\partial S|/|S|\geq \xi\frac{ \log(1/\alpha)}{(\log n)^2} \mbox{ for any set $S$ with $|S|=\alpha N$ and $\alpha\in (0,99/100)$.}
\end{equation}
To prove \eqref{eq:finedexge44g4}, let $S':=\boxcore(S)$. We have the following case analysis.
\begin{description}
\item[Case I.] Suppose that $|S'|< \tfrac{1}{2}|S|$. Then,  we have by \eqref{eq:torusdecom4t3} (with $\eta=1/2$) that 
\[|\partial S|\geq |\partial^* S|\geq \frac{|S|}{8\ell^2}\geq \frac{|S|}{32\ell_0^2\log n},\]
and therefore  \eqref{eq:finedexge44g4} holds (using that $3\log n\geq \log(1/\alpha)$).

\item[Case II.]Suppose that  $|S'|\geq \tfrac{1}{2}|S|$. Let $\alpha'=|S'|/N$ and note that $\alpha'\leq \alpha$. By Item~\ref{it:eeff}, we have 
\[|\partialv S'|\geq c|S'|\frac{\log(1/\alpha')}{\log n}\geq \frac{c|S|}{2}\frac{\log(1/\alpha)}{\log n}.\]
 
If $|(\partialv S')\backslash S|\geq \tfrac{1}{2}|\partialv S'|$, then we have that
\[|\partial S|\geq |(\partialv S')\backslash S|\geq \frac{1}{2}|\partialv S'|\geq   \frac{c|S|}{4}\frac{\log(1/\alpha)}{\log n},\]
and therefore  \eqref{eq:finedexge44g4} holds.

If $|(\partialv S')\backslash S|< \tfrac{1}{2}|\partialv S'|$, then we have that $|(\partialv S')\cap S|>\tfrac{1}{2}|\partialv S'|$. Note that $(\partialv S')\cap S\subseteq S\backslash S'$ and hence we conclude that there are more than $\frac{c|S|}{4}\frac{\log(1/\alpha)}{\log n}$ vertices in $S$ which do not belong to the box-core of $S$. Then,  we have by \eqref{eq:torusdecom4t3} (with $\eta=\frac{c \log(1/\alpha)}{4\log n}$) that 
\[|\partial S|\geq |\partial^* S|\geq\frac{c \log(1/\alpha)}{4\log n}\frac{|S|}{4\ell^2},\]  
and therefore  \eqref{eq:finedexge44g4} holds.
\end{description}

We are now ready to bound the conductance $\Phi(S)$ for a connected set $S$ with $\pi(S)\leq 1/2$ (this part of the proof is analogous to the corresponding part in Lemma~\ref{lem:conrltwo}). The assumption $\pi(S)\leq 1/2$ gives that $\sum_{v\in S}d_v\leq |E|\leq 3N$, where the last inequality holds from Item~\ref{it:edges1}. Since $d_v\geq 4$ for any $v\in V$, we have that $|S|\leq 3N/4$. Moreover, we have that $\Phi(S)\geq \frac{|\partial S|}{\sum_{v\in S}d_v}$, see \eqref{eq:phisrl2} for details.

For any connected set $S\subseteq V$ we have by Item~\ref{it:avgd435} that
\begin{equation*}
\sum_{v\in S}d_v\leq M\max\{|S|,\log n\}.
\end{equation*}
Combining this with \eqref{eq:finedexge44g4} yields that $\Phi(S)\geq \frac{\xi}{M}\frac{ \log(1/\alpha)}{(\log n)^2}$ for all connected sets $|S|$ with $|S|>\log n$ and $|S|=\alpha N$. For connected sets $|S|$ with $|S|\leq\log n$ and $|S|=\alpha N$, we have by Theorem~\ref{thm:torusexpansion} that $|\partial S|\geq |\partial^* S|\geq 2\sqrt{|S|}$ and hence $\Phi(S)\geq \frac{2\sqrt{|S|}}{M\log n}\geq \frac{2}{3M}\frac{ \log(1/\alpha)}{(\log n)^2}$ (using that $3\log n\geq \log(1/\alpha)$).

This concludes the proof of Lemma~\ref{lem:conreqtwo} with the constant $\tau:=\min\{\xi/M,2/(3M)\}>0$.
\end{proof}

\section{Lower bounds on the mixing time}\label{sec:lowout}
The lower bounds in Theorem~\ref{thm:main} follow from the following lower bounds on the mixing time. 
\begin{theorem}\label{thm:main2}
Let $r\geq 0$ and $n$ be sufficiently large. Then, with probability $1-O(1/n)$ over the choice of the graph $G\sim\Gc_{n,r}$, the lazy random walk on $G$ satisfies
\[\Tmix=\begin{cases} \Omega(\log n) &\mbox{for }r\leq 2\\ \Omega(n^{r-2}) &\mbox{for }2<r<3\\ \Omega(n/\log n) &\mbox{for }r=3\\ \Omega(n) &\mbox{for }r>3.\end{cases}\]
\end{theorem}

Our proof of Theorem \ref{thm:main2} upper bounds the 
conductance of the random walk
which implies the same lower bound as in Theorem~\ref{thm:main} 
on the relaxation time (i.e., inverse spectral gap).  For $2<r<4$, 
Crawford and Sly~\cite{CS1} proved, for the random walk on the infinite cluster of
the long-range percolation (LRP) process the relaxation time is $n^{r-2}(\log{n})^{O(1)}$.
Thus, these bounds on the relaxation time are tight within log factors for $2<r<3$. 
The gap in the range $3<r<4$ is an open question, see Section~\ref{sec:discussion}.
For $d=1$ analogous lower bounds were proved for LRP in Benjamini, Berger, and 
Yadin~\cite{ben2008}.

For $r>4$ we conjecture that our model is ``fully local'' as is the
case for LRP.  We expect that one can obtain an improved lower bound of
$\Omega(n^2/\log{n})$ following the approach of \cite{ben2008}:
proving a linear lower bound on the diameter 
as is done by Berger~\cite{NBerger}
for LRP and then applying the Varopoulos-Carne bound~\cite{Varopoulos,Carne} (see also \cite[Proposition 13.11]{Lyons}).

We next give the proof of Theorem~\ref{thm:main2} (which follows  by combining the upcoming Corollaries~\ref{cor:Tmixrg2} and~\ref{cor:Tmixrl2}).

\subsection{Lower bounds for $r>2$}\label{sec:lowergthan2}
In this section, we prove the lower bounds on the mixing time $\Tmix$ for $r>2$ stated in 
Theorem~\ref{thm:main2} for the small-world network.  
Our goal will be to show a set $S$ with small conductance. Let $G=(V,E)\sim \Gc_{n,r}$.
 
Set $L= cn$ for a constant $c\in(0,1]$; later, we will set $c=9/10$. Let $S$ be the set of vertices within distance $\leq L$ from the origin $\rho=(0,0)$. Since $L\leq n$, we have that $|S|=1+4\sum^L_{\ell=1}\ell=2L^2+\Theta(L)$. 
\begin{lemma}\label{lem:torpidconductance}
Let $r>2$ and $n$ be sufficiently large. For any constant $c\in(0,1]$, with probability $1-\exp(-\Omega(n))$ over the choice of the graph $G$, it holds that
\begin{enumerate}
\item $\sum_{v\in S}d_v\leq 8|S|$.
\item $|\partial S| = O(L^{4-r})$  for $2<r<3$, $|\partial S|=O(L\log{L})$ for $r= 3$, and  $|\partial S|=O(L)$ for $r>3$.
\end{enumerate}
\end{lemma}
  
\begin{proof}
To simplify notation, all probability bounds in the proof are with respect to the choice of the random graph $G\sim \Gc_{n,r}$. For vertices $u,w\in V$,  let $Y_{u,w}$ be the indicator r.v. that there is a long-range edge from $u$ to $w$. Let $\binom{S}{2}:=\big\{\{u,w\}\mid u,w\in S, u\neq w\big\}$  denote the set of all unordered pairs of vertices in $S$ and define
\[Y_{in}:=\sum_{\{u,w\}\in \binom{S}{2}} Y_{u,w}, \quad Y_{out}:=\sum_{u\in S, w\in V\backslash S}Y_{u,w}.\] 
We have
\begin{equation}\label{eq:upperXXm}
\sum_{v\in S}d_v\leq 4|S|+ 2Y_{in}+Y_{out}, \qquad |\partial S|\leq 16L+Y_{out}.
\end{equation}
To see the first inequality in \eqref{eq:upperXXm}, note that in the sum $\sum_{v\in S}d_v$, the edges of the torus contribute $4|S|$ and the long-range edges contribute $2Y_{in}+Y_{out}$. The second inequality in \eqref{eq:upperXXm} follows similarly, by noting now that the edges of the torus that contribute to $|\partial S|$ must be incident to a vertex at distance $L$ from the origin; the number of such vertices is $4L$ and each of these vertices can (crudely) contribute at most four edges to $|\partial S|$.

Note that for every $u,w\in V$, it holds that $\Pr[Y_{u,w}=1]\leq \frac{1}{Z}\dist{u}{w}^{-r}$. Further, the random variables $\{Y_{u,w}\}_{u,w\in V}$ are independent. Thus, with $\mu_{in}:=\mathbf{E}[Y_{in}]$ and $\mu_{out}:=\mathbf{E}[Y_{out}]$, we obtain by Lemma~\ref{lem:chernoff} that for any $t\geq0$, it holds that
\begin{equation}\label{eq:upperboundS}
\begin{gathered}
{\Pr}(Y_{in}\geq \mu_{in}+t)\leq \exp\Big(-\frac{t^2}{2(\mu_{in}+t/3)}\Big), \quad
{\Pr}(Y_{out}\geq \mu_{out}+t)\leq \exp\Big(-\frac{t^2}{2(\mu_{out}+t/3)}\Big),
\end{gathered}
\end{equation}
We therefore focus on obtaining upper bounds for $\mu_{in},\mu_{out}$.

 To bound $\sum_{v\in S}d_v$, we will need only crude bounds on $\mu_{in}$ and $\mu_{out}$. In particular, since every vertex $u\in S$ adds in expectation one long-range neighbour, we have that $\mu_{in}\leq |S|$ and $\mu_{out}\leq |S|$. Thus, using \eqref{eq:upperboundS} with $t=|S|/3$, a union bound gives that 
\begin{equation}\label{eq:v34frv38678}
\mbox{ with probability $1-\exp(-\Omega(n^2))$, it holds that $2Y_{in}+Y_{out}\leq 4|S|$}.
\end{equation}

To obtain a more precise bound for $|\partial S|$, we will need a more refined upper bound for  $\mu_{out}$. Let $u\in S$ be at distance $i>0$ from the origin $\rho$, so that $i\leq L$. Clearly, all vertices within distance $L-i$ from $u$ belong to $S$ and hence every vertex $w\in V\backslash S$ must satisfy $\dist{u}{w}\geq L-i+1$. Note that the number of vertices at distance $\ell$ from $u$ is $4\min\{\ell,2n+1-\ell\}\leq 4\ell$, and hence we obtain the bound
\[\sum_{w\in V\backslash S}\Pr[Y_{u,w}=1]\leq \frac{4}{Z}\sum^{2n}_{\ell=L-i+1}\frac{1}{\ell^{r-1}}.\]
Since there are $4i$ vertices at distance $i>0$ from the origin, we obtain (by using the trivial bound that the origin has in expectation one long-range edge incident to it) that 
\[\mu_{out}\leq  1+\frac{16}{Z}\sum^{L}_{i=1} \sum^{2n}_{\ell=L-i+1}\frac{i}{\ell^{r-1}}=1+\frac{16}{Z} \left(\sum^{L}_{\ell=1}\sum^{L}_{i=L-\ell+1}\frac{i}{\ell^{r-1}}+\sum^{2n}_{\ell=L+1}\sum^{L}_{i=1}\frac{i}{\ell^{r-1}}\right),\] 
where the last equality follows by switching the order of summation. Since $r>2$, by Lemma~\ref{lem:factor} we have that $Z=\Theta(1)$. Note also that 
\[\sum^{L}_{\ell=1}\sum^{L}_{i=L-\ell+1}\frac{i}{\ell^{r-1}}\leq \sum^{L}_{\ell=1}\frac{L}{\ell^{r-2}}=\begin{cases}
    O(L^{4-r})  & \mbox{ for } 2<r<3 \\
    O(L\log{L})  & \mbox{ for } r= 3 \\
    O(L)  & \mbox{ for } r>3, \end{cases}\]
and, for all $r>2$,
\[\sum^{2n}_{\ell=L+1}\sum^{L}_{i=1}\frac{i}{\ell^{r-1}}\leq \sum^{2n}_{\ell=L+1}\frac{L^2}{\ell^{r-1}}\leq (2n-L)L^{3-r}=O(L^{4-r}).\]
It follows that 
\[\mu_{out}= \begin{cases}
    O(L^{4-r})  & \mbox{ for } 2<r<3 \\
    O(L\log{L})  & \mbox{ for } r= 3 \\
    O(L)  & \mbox{ for } r>3. \end{cases}\]
Using that $L=\Omega(n)$,  we obtain by \eqref{eq:upperboundS} that 
\begin{equation}\label{eq:rv34edh6}
\mbox{with probability $1-\exp(-\Omega(n))$, it holds than $Y_{out}\leq 2\mu_{out}$}.
\end{equation} 
Combining \eqref{eq:v34frv38678} and \eqref{eq:rv34edh6} by a union bound, and  plugging them in \eqref{eq:upperXXm}, we obtain the statement of the lemma. 
\end{proof}
We are now ready to obtain lower bounds on the mixing time for $r>2$.
\begin{corollary}\label{cor:Tmixrg2}
Let $r>2$ and $n$ be sufficiently large. Then, with probability $1-\exp(-\Omega(n))$ over the choice of the graph of the graph $G\sim\Gc_{n,r}$, the lazy random walk on $G$ satisfies
\[\Tmix=\begin{cases} \Omega(n^{r-2}) &\mbox{for }2<r<3\\ \Omega(n/\log n) &\mbox{for }r=3\\ \Omega(n) &\mbox{for }r>3.\end{cases}\]
\end{corollary}
\begin{proof}
Let $L=cn$, where $c=9/10$. Let $G\sim \Gc_{n,r}$ and, as before, let $S$ be the set of vertices at distance $\leq L$ from the origin. Recall that $|S|=2L^2+\Theta(L)=2c^2n^2+\Theta(n)=\frac{1}{2}c^2N+o(N)$, where $N=(2n+1)^2$ is the number of vertices in $G$.

By a union bound over the events in Lemmas~\ref{lem:numedges} and~\ref{lem:torpidconductance}, we have  with probability $1-\exp(-\Omega(n))$ that $2N\leq |E|\leq 3N$ and that the set $S$ satisfies 
\begin{enumerate}
\item \label{it:mass}$\sum_{v\in S}d_v\leq 8|S|$.
\item \label{it:cutedges}$|\partial S| = O(L^{4-r})$  for $2<r<3$, $|\partial S|=O(L\log{L})$ for $r= 3$, and  $|\partial S|=O(L)$ for $r>3$.
\end{enumerate}
By Item~\ref{it:mass} and the choice of the constant $c$, we have that 
\[\frac{\sum_{v\notin S}d_v}{2|E|}\geq\frac{2|E|-8|S|}{2|E|}\geq 1-\frac{2|S|}{|N|}\geq 1/10\] for all sufficiently large $n$. Thus, using Item~\ref{it:cutedges} and the (deterministic) bound $\sum_{v\in S}d_v\geq 4|S|$, we obtain that the conductance of the set $S$ is bounded by
\[\Phi(S)=\frac{|\partial S|}{\frac{1}{2|E|}\big(\sum_{v\in S}d_v\big)\big(\sum_{v\notin S}d_v\big)}=\begin{cases} O(n^{2-r})&\mbox{for } 2<r<3\\ O(\log n/n)&\mbox{for } r=3\\ O(1/n)&\mbox{for } r>3.\end{cases}\]
It follows that the conductance $\Phi$ of the simple random walk on $G$ satisfies $\Phi\leq \Phi(S)$, and hence the corollary follows from Theorem~\ref{lem:mixingtime}.
\end{proof}

\subsection{Simple lower bound for $r\leq 2$ via diameter}\label{sec:rlessthan2}

Our lower bounds on the mixing time for $r\in [0,2]$ are based on the following well-known lemma. For a (connected) graph $G$, we denote by $\diam(G)$ the largest distance between any two vertices in~$G$.
\begin{lemma}\label{lem:Tmixgeneral}
Let $G=(V,E)$ be a connected undirected graph. Then, the mixing time $\Tmix$ of the lazy random walk on $G$ satisfies $\Tmix\geq \diam(G)/3$.
\end{lemma}
\begin{proof}
Let $L=\diam(G)$ and let $T=\left\lfloor (L-1)/2\right\rfloor$. Consider two vertices $u,v$ whose graph distance is $L$. Then, observe that the vectors $P^{T}(u,\cdot)$ and $P^{T}(v,\cdot)$ are supported on disjoint subsets of $V$ and hence
\[\norm{P^{T}(u,\cdot)-P^{T}(v,\cdot)}_{\mathrm{TV}}=1.\]
By the triangle inequality, we have that either $\norm{P^{T}(u,\cdot)-\pib}_{\mathrm{TV}}\geq 1/2$ or $\norm{P^{T}(v,\cdot)-\pib}_{\mathrm{TV}}\geq 1/2$, thus proving that $\Tmix\geq T\geq L/3$.
\end{proof}

\begin{lemma}\label{lem:diambounds}
Let $r\in[0,2]$ and $n$ be a sufficiently large positive integer. Then, with probability $1-O(1/n)$ over the choice of the graph $G\sim \Gc_{n,r}$, it holds that $\diam(G)=\Omega(\log n)$.
\end{lemma}
\begin{proof}
Consider an arbitrary vertex $v$ in $G$. Since the expected degree of every vertex in $G$ is equal to $5$, we have that the expected number of vertices at distance $\leq k$ from $v$ is at most $5^{k+1}$. Applying Markov's inequality for $k=\frac{1}{100}\left\lfloor \log n\right\rfloor$, we obtain that with probability $1-O(1/n)$ over the choice of the graph $G$ there are at most $n^2/2<N$ vertices at distance $\leq k$ from $v$. For all such graphs $G$, we clearly have that $\diam(G)=\Omega(\log n)$, as wanted.
\end{proof}
Combining Lemmas~\ref{lem:Tmixgeneral} and~\ref{lem:diambounds}, we immediately obtain the following corollary for all $r\in [0,2]$.
\begin{corollary}\label{cor:Tmixrl2}
Let $r\in[0,2]$ and $n$ be a positive integer. With probability $1-O(1/n)$ over the choice of the graph $G\sim\Gc_{n,r}$, the lazy random walk on $G$ satisfies
$\Tmix=\Omega(\log n)$.
\end{corollary}

\section{Discussion}\label{sec:discussion}

Kleinberg's routing results generalize to the $d$-dimensional integer lattice for $d\geq 1$, showing
that there is fast routing when $r=d$ and it is exponentially slower for $r\neq d$.
The techniques in this paper can likely be applied to show a phase transition on the mixing time of the lazy random walk at $r=d$. Similarly, we expect that our techniques also apply to the ``1-out'' version of Kleinberg's model, in which each vertex selects a (directed) long-range edge with appropriate probability and then the direction of the edge is ``forgotten".

There are several intriguing directions to pursue:
\begin{itemize}
\item
Can one establish matching upper and lower bounds for $r=2$? We conjecture that the mixing time is $\Theta((\log n)^2)$. Recall, our result for $r=2$ is an upper bound of $O((\log n)^4)$ and a lower bound of $\Omega(\log n)$.
Achieving a lower bound of $\omega(\log{n})$ for $r=2$ would be quite interesting as it
would establish separation between the $r<2$ and $r=2$ cases.
\item 
For $2<r<3$ our lower bound matches (within poly-log factors) the upper bound 
of Crawford and Sly~\cite{CS1} 
for long-range percolation, as discussed in Section~\ref{sec:lowout}.
However there is a gap for $3<r<4$ as we prove a lower bound of 
$\Omega(n)$ and Crawford-Sly prove $n^{r-2}(\log n)^{O(1)}$.
It is an interesting open question to determine the correct polynomial for this case 
$3\leq r<4$.

\item
It would be very interesting to analyze the {\em directed} version of Kleinberg's small-world network.
One of the challenges here is to understand the stationary distribution of the random walk since
it is no longer reversible. Recent progress on analysing the mixing time of random walks on random sparse digraphs has been made in \cite{BCS,BCS1} (see also \cite{CF}).  
A fascinating open problem for the directed model 
is to determine whether the stationary distribution
is nearly-uniform, i.e., whether it is within $\poly$-$\log$ factors of uniform, for all $r$.  
\end{itemize}

\bibliographystyle{plain}

\end{document}